\newif\iffull

\fullfalse

\documentclass[11pt]{article}

\usepackage{hyperref}
\usepackage{setspace}
\usepackage[\iffull notes=true, \else notes=false, \fi later=false]{dtrt}
\usepackage{fullpage}
\usepackage{microtype}
\usepackage{amssymb,amsfonts,amsmath,amsthm}
\usepackage{mathtools}
\usepackage{fullpage}
\usepackage{graphicx}
\usepackage{verbatim}
\usepackage{array}
\usepackage{latexsym}
\usepackage{paralist}
\usepackage{enumitem}
  \setlist[itemize]{leftmargin=*}
  \setlist[enumerate]{leftmargin=*}
\usepackage{makecell}
\usepackage{hypcap}
\usepackage{mleftright}
\usepackage{colortbl}
\usepackage{mathrsfs}
\usepackage{microtype}
\usepackage{etoolbox}
\usepackage{graphicx}
\usepackage{breakcites}
\usepackage{amssymb}
\usepackage{multirow}
\usepackage[small,margin=2mm]{caption}
\usepackage[capitalize,nameinlink]{cleveref}
\usepackage[vlined,linesnumbered,ruled,resetcount]{algorithm2e}
\usepackage{algpseudocode}
\usepackage{float}
\usepackage{setspace}
\usepackage[utf8]{inputenc}
\usepackage{amsthm}
\usepackage{amsmath}
\usepackage{times} %
\usepackage{comment}
\usepackage{xspace}

\crefname{algocf}{Algorithm}{Algorithms}
\Crefname{algocf}{Algorithm}{Algorithms}

\newcommand{\doclearpage}{%
\iffull
\clearpage
\else
\fi
}

\newcommand{\keywords}[1]{\bigskip\par\noindent{\footnotesize\textbf{Keywords\/}: #1}}

\newcommand{\defemph}[1]{\textbf{\emph{#1}}}

\newcommand{\FormatAuthor}[3]{
\begin{tabular}{c}
#1 \\ {\small\texttt{#2}} \\ {\small #3}
\end{tabular}
}

\theoremstyle{plain} 
\newtheorem{theorem}{Theorem}[section]
\newtheorem{lemma}[theorem]{Lemma}

\newtheorem{claim}[theorem]{Claim}

\newtheorem{definition}[theorem]{Definition}

\newtheorem{observation}[theorem]{Observation}

\newtheorem{notation}[theorem]{Notation}
\newtheorem{stheorem}{Theorem}

\theoremstyle{definition} 
\newtheorem{remark}[theorem]{Remark}

\theoremstyle{remark} 

\DeclareMathOperator{\poly}{poly}
\DeclareMathOperator{\polylog}{polylog}

\DeclareMathOperator{\E}{\mathbb{E}}
\DeclareMathOperator{\Var}{\mathbf{Var}}

\newcommand{\seq}{\subseteq}

\newcommand{\eps}{\epsilon}
\newcommand{\abs}[1]{|#1|}
\newcommand{\dist}{\mathrm{dist}}

\renewcommand{\P}{\mathcal{P}}
\newcommand{\C}{\mathcal{C}}

\newcommand{\coord}{k}

\newcommand{\Bits}{\{0,1\}}
\newcommand{\N}{\mathbb{N}}
\newcommand{\Field}{\mathbb{F}}
\newcommand{\F}{\Field}
\newcommand{\SubF}{\mathbb{H}}
\renewcommand{\H}{{\SubF}}

\newcommand{\RM}{\mathsf{RM}}
\renewcommand{\deg}{d}

\newcommand{\DefineEqual}{\coloneqq}

\newcommand{\floor}[1]{\lfloor #1 \rfloor}
\newcommand{\ceil}[1]{\lceil #1 \rceil}

\newcommand{\FormatComplexityClass}[1]{\mathbf{#1}}
\newcommand{\PClass}{\FormatComplexityClass{P}}

\newcommand{\Proof}{\pi}

\newcommand{\Code}{C}
\renewcommand{\message}{M}
\newcommand{\inputword}{w}
\newcommand{\codeword}{Q}
\newcommand{\ClosestCodeword}{W^*}
\newcommand{\Word}{f}

\newcommand{\BaseCode}{{\RM(\RMDim, \RMDeg)}}

\newcommand{\TensorDim}{\RMDim}

\newcommand{\ComposedCode}{\Code_{comp}}

\newcommand{\CodeDim}{K}
\newcommand{\BlockLength}{N}
\newcommand{\RSdeg}{d}
\newcommand{\RSlen}{n}
\newcommand{\CodeLen}{n}
\newcommand{\Decoder}{{\mathcal D}}
\newcommand{\DistCode}{\delta}
\newcommand{\Point}{\vecx}
\newcommand{\Line}[2]{{\ell_{#1, #2}}}
\newcommand{\Plane}[3]{{\P_{#1, #2, #3}}}
\newcommand{\Index}{\vecx}

\newcommand{\Language}{L}

\newcommand{\Robust}{\rho}
\newcommand{\Soundness}{\epsilon}

\newcommand{\RMDist}{\delta_{\RM}}
\newcommand{\DistRM}{\RMDist}
\newcommand{\DecRad}{{\tau_{\rm cor}}}
\newcommand{\TestRadius}{\tau}
\newcommand{\Randomness}{r}
\newcommand{\Query}{q}

\newcommand{\QueryRLCC}{\Query_{RLCC}}
\newcommand{\SoundnessRW}{\Soundness_{RW}}
\newcommand{\SoundnessPCPP}{\Soundness_{PCPP}}
\newcommand{\QueryPCPP}{\Query_{\rm PCPP}}
\newcommand{\SoundnessInnerRLCC}{\Soundness_{inRLCC}}
\newcommand{\SoundnessRLCC}{\Soundness_{RLCC}}

\newcommand{\CTRW}{{\sf CTRW}}
\newcommand{\RWSteps}{r}

\newcommand{\Wclose}{W_{\rm close}}

\newcommand{\textdef}[1]{\emph{#1}}
\newcommand{\vecx}{{\vec{x}}}
\newcommand{\vecy}{{\vec{y}}}
\newcommand{\vecz}{{\vec{z}}}
\newcommand{\veca}{{\vec{a}}}

\newcommand{\vech}{{\vec{h}}}

\newcommand{\MessLen}{k}
\newcommand{\RMDeg}{d}
\newcommand{\RMLen}{n}

\newcommand{\RMDim}{m}
\newcommand{\ProofPart}{\Pi}
\newcommand{\NumRep}{t}
\newcommand{\RMPart}{{\RM}^{rep}}

\newcommand{\PCPPLen}{len}
\newcommand{\Subroutine}{\CTRW-Test}

\begin{document}

\title{Relaxed Locally Correctable Codes with Improved Parameters}

\author{
\begin{tabular}[h!]{ccc}
 & \FormatAuthor{Vahid R. Asadi}{vasadi@sfu.ca}{Simon Fraser University}
 & \FormatAuthor{Igor Shinkar}{ishinkar@sfu.ca}{Simon Fraser University}
\end{tabular}
}

\iffull
  \date{\today}
\else
  \date{}
\fi

\maketitle

\begin{abstract}

Locally decodable codes (LDCs) are error-correcting codes $\Code \colon \Sigma^\MessLen \to \Sigma^\CodeLen$
that admit a local decoding algorithm that recovers each individual bit of the message
by querying only a few bits from a noisy codeword.
An important question in this line of research is to understand the optimal trade-off between the query complexity of LDCs and their block length.
Despite importance of these objects, the best known constructions of constant query LDCs have super-polynomial length,
and there is a significant gap between the best constructions and the known lower bounds in terms of the block length.

For many applications it suffices to consider the weaker notion of \emph{relaxed LDCs (RLDCs)},
which allows the local decoding algorithm to abort if by querying a few bits it detects that the input is not a codeword.
This relaxation turned out to allow decoding algorithms with constant query complexity for codes with \emph{almost linear} length.
Specifically, \cite{BGHSV06} constructed an $O(\Query)$-query RLDC that encodes a message of length $\MessLen$ using a codeword of block length $\CodeLen = O(\MessLen^{1+1/\sqrt{\Query}})$.

In this work we improve the parameters of \cite{BGHSV06} by constructing an $O(\Query)$-query RLDC
that encodes a message of length $\MessLen$ using a codeword of block length $O(\MessLen^{1+1/{\Query}})$.
This construction matches (up to a multiplicative constant factor) the lower bounds of \cite{KatzTrevisan00, Woodruff07} for constant query \emph{LDCs},
thus making progress toward understanding the gap between LDCs and RLDCs in the constant query regime.

In fact, our construction extends to the stronger notion of relaxed locally \emph{correctable} codes (RLCCs), introduced in \cite{GurRR18},
where given a noisy codeword the correcting algorithm either recovers each individual bit of the codeword
by only reading a small part of the input, or aborts if the input is detected to be corrupt.


\end{abstract}

\clearpage
\setcounter{tocdepth}{2}
{\footnotesize \tableofcontents}

\clearpage
\section{Introduction}
\label{sec:intro}
\emph{Locally decodable codes (LDCs)} are error-correcting codes that admit a decoding algorithm that recovers each specific symbol of the message by reading a small number of locations in a possibly corrupted codeword. More precisely, a locally decodable code $\Code \colon \F^\MessLen \to \F^\CodeLen$ with local decoding radius $\TestRadius \in [0,1]$ is an error-correcting code that admits a local decoding algorithm $\Decoder_{\Code}$, such that given an index $i \in [\MessLen]$ and a corrupted word $w \in \F^\CodeLen$ which is $\TestRadius$-close to an encoding of some message $\Code(\message)$, reads a small number of symbols from $w$, and outputs $\message_{i}$ with high probability.
Similarly, we have the notion of \emph{locally correctable codes (LCCs)}, which are error-correcting codes that not only admit a local algorithm that decode each symbol of the message, but are also required to correct an arbitrary symbol from the entire codeword.
Locally decodable and locally correctable codes have many applications in different areas of theoretical computer science, such as complexity theory, coding theory, property testing, cryptography, and construction of probabilistically checkable proof systems. For details, see the surveys~\cite{Yekhanin12,KS17} and the references within.

Despite the importance of LDCs and LCCs, and the extensive amount of research studying these objects, the best known construction of constant query LDCs has super-polynomial length $\CodeLen = \exp(\exp(\log^{\Omega(1)}(\MessLen)))$, which is achieved by the highly non-trivial constructions of \cite{Yekhanin08} and \cite{Efremenko12}. For constant query LCCs, the best known constructions are of exponential length, which can be achieved by some parameterization of Reed-Muller codes.
It is important to note that there is huge gap between the best known lower bounds for the length of constant query LDCs and the length of best known constructions. Currently, the best known lower bound on the length of LDCs says that for $\Query \geq 3$ it must be at least $\MessLen^{1 + \Omega(1/\Query)}$, where $\Query$ stands for the query complexity of the local decoder. See \cite{KatzTrevisan00,Woodruff07} for the best general lower bounds for constant query LDCs.

Motivated by applications to probabilistically checkable proofs (PCPs), Ben-Sasson, Goldreich, Harsha, Sudan, and Vadhan introduced in \cite{BGHSV06} the notion of \emph{relaxed locally decodable codes (RLDCs)}.
Informally speaking, a relaxed locally decodable code is an error-correcting code which allows the local decoding algorithm to abort if the input codeword is corrupt, but does not allow it to err with high probability.
In particular, the decoding algorithm should always output correct symbol, if the given word is not corrupted.
Formally, a code $\Code \colon \F^\MessLen \to \F^\CodeLen$ is an RLDC with decoding radius $\TestRadius \in [0,1]$
if it admits a relaxed local decoding algorithm $\Decoder_{\Code}$
which given an index $i \in [\MessLen]$ and a possibly corrupted codeword $w \in \F^\CodeLen$, makes a small number of queries to $w$, and satisfies the following properties.
\begin{description}
  \item[Completeness:] If $w = \Code(\message)$ for some $\message \in \F^\MessLen$, then $\Decoder_{\Code}^{w}(i)$ should output $\message_i$.
  \item[Relaxed decoding:] If $w$ is $\TestRadius$-close to some codeword $\Code(\message)$,
  then $\Decoder_{\Code}^{w}(i)$ should output either $\message_i$ or a special \emph{abort} symbol with probability at least 2/3.
\end{description}
This relaxation turns out to be very helpful in terms of constructing RLDCs with better block length.
Indeed, \cite{BGHSV06} constructed of a $\Query$-query RLDC with block length $\CodeLen = \MessLen^{1+O(1/\sqrt{\Query})}$.

\medskip
The notion of \emph{relaxed LCCs (RLCCs)}, recently introduced in \cite{GurRR18}, naturally extends the notion of RLDCs.
These are error-correcting codes that admit a correcting algorithm that is required to correct every symbol of the codeword,
but is allowed to abort if noticing that the given word is corrupt.
More formally, the local correcting algorithm gets an index $i \in [\CodeLen]$,
and a (possibly corrupted) word $w \in \F^\CodeLen$, makes a small number of queries to $w$,
and satisfies the following properties.
\begin{description}
  \item[Completeness:] If $w \in \Code$, then $\Decoder_{\Code}^{w}(i)$ should output $w_i$.
  \item[Relaxed correcting:] If $w$ is $\TestRadius$-close to some codeword $c^* \in \Code$,
  then $\Decoder_{\Code}^{w}(i)$ should output either $c^*_i$ or a special \emph{abort} symbol with probability at least 2/3.
\end{description}
Note that if the code $\Code$ is systematic, i.e., the encoding of any message $\message \in \F^\MessLen$
contains $\message$ in its first $\MessLen$ symbols, then the notion of RLCC is stronger than RLDC.

Recently, building on the ideas from~\cite{GurRR18}, \cite{CGS20} constructed RLCCs whose block length matches the RLDC construction
of \cite{BGHSV06}.
For the lower bounds, the only result we are aware of is the work of Gur and Lachish \cite{GurL19},
who proved that for any RLDC the block length must be at least $\CodeLen = \MessLen^{1+\Omega(1/{\Query^2})}$.

\medskip
Given the gap between the best constructions and the known lower bounds, it is natural to ask the following question:

\medskip
\fbox{
What is the best possible trade-off between the query complexity and the block length of an RLDC?
}
\medskip

In particular, \cite{BGHSV06} asked whether it is possible to obtain a $\Query$-query RLDC whose
block length is strictly smaller than the best known lower bound on the length of LDCs.
A positive answer to their question would show a separation between the two notions,
thus proving that the relaxation is \emph{strict}.
See paragraph \emph{Open Problem} in the end of Section 4.2 of \cite{BGHSV06}.

In this work we make progress on this problem by constructing a relaxed locally decodable code
$\Code \colon \Field^\CodeDim \to \Field^\BlockLength$ with
query complexity $O(\Query)$ and block length $\CodeDim^{1+O(1/\Query)}$.
In fact, our construction gives the stronger notion of a relaxed locally correctable code.

\begin{stheorem}[Main Theorem]\label{thm:main}
For every $\Query \in \N$ there exists an $O(\Query)$-query relaxed locally correctable code
$\Code \colon \Bits^\CodeDim \to \Bits^\BlockLength$ with constant relative distance and constant decoding radius,
such that the block length of $\Code$ is
\begin{equation*}
	\BlockLength=\Query^{O(\Query^2)} \cdot \CodeDim^{1+O(1/\Query)} \enspace.
\end{equation*}
\end{stheorem}

\noindent
Therefore, our construction improves the parameters of the $O(\Query)$-query RLDC
construction of \cite{BGHSV06} with block length $\BlockLength=\CodeDim^{1+O(\sqrt{1/\Query})}$,
and matches (up to a multiplicative factor in $\Query$) the lower bound of $\Omega(\CodeDim^{1+\frac{1}{\ceil{\Query/2}-1}})$
for the block length of $\Query$-query LDCs~\cite{KatzTrevisan00, Woodruff07}.

\begin{remark}
    In this paper we prove \cref{thm:main} for a code $\Code \colon \Field^\CodeDim \to \Field^\BlockLength$ over a large alphabet.
    Specifically, we show a code $\Code \colon \Field^\CodeDim \to \Field^\BlockLength$
    satisfying \cref{thm:main}, for a finite field $\Field$ satisfying $\abs{\Field} \geq c_q \cdot \CodeDim^{1/\Query}$,
    for some $c_\Query \in \N$ that depends only on $\Query$.

    Using the techniques from \cite{CGS20} it is not difficult to obtain an
    RLCC over the binary alphabet with almost the same block length.
    Indeed, this can be done by concatenating our code over large alphabet
    with an arbitrary binary code with constant rate and constant relative distance.
    See \cref{sec:conclusions-open-problems} for details.
\end{remark}


\subsection{Related works}
\parhead{RLDC and RLCC constructions:} Relaxed locally decodable codes, were first introduced by \cite{BGHSV06},
motivated by applications to constructing short PCPs.
Their construction has a block length equal to $\BlockLength = \CodeDim^{1 + O(1/\sqrt{\Query})}$.
Since that work, there were no constructions with better block length, in the constant query complexity regime   .
Recently, \cite{GurRR18} introduced the related notion of relaxed locally correctable codes (RLCCs),
and constructed $\Query$-query RLCCs with block length $\BlockLength = \poly(\CodeDim)$.
Then, \cite{CGS20} constructed relaxed locally correctable codes with block length
matching that of \cite{BGHSV06} (up to a multiplicative constant factor $\Query$).
The construction of \cite{CGS20} had two main components, that we also use in the current work.
\begin{description}
  \item[Consistency test using random walk ($\CTRW$):] Informally, given a word $w$, and a coordinate $i$ we wish to correct,
  $\CTRW$ samples a sequence of constraints $\C_1,\C_2,\dots,\C_t$ on $w$, such that the domains of $\C_i$ and $\C_{i+1}$ intersect, with the guarantee that
  if $w$ is close to some codeword $c^* \in \Code$, but $w_i \neq c^*_i$, then with high probability $w$
  will be far from satisfying at least one of the constraints.
  In other words, $\CTRW$ performs a random walk on the constraints graph and checks if $w$ is consistent with $c^*$ in the $i$'th coordinate.
  We introduce this notion in detail in \cref{sec:overview-CTRW},
  and prove that the Reed-Muller code admits a $\CTRW$ in \cref{sec:CTRW-RM}.

  \item[Correctable canonical PCPPs (ccPCPP):] These are PCPP systems for some specified language $\Language$
  satisfying the following properties:
  \begin{inparaenum}[(i)]
  \item for each $w \in \Language$ there is a unique proof $\pi(w)$ that satisfies the verifier with probability 1,
  \item the verifier accepts with high probability only pairs $(x, \pi)$ that are
  close to some $(w, \pi(w))$ for some $w \in \Language$, i.e., only the pairs where $x$ is close to some $w \in \Language$,
  and $\pi$ is close to $\pi(w)$, and
  \item the set $\{w \circ \pi_w : w \in \Language\}$ is an RLCC.
\end{inparaenum}
Canonical proofs of proximity have been studies in~\cite{DinurGG18, Paradise20}.
We elaborate on these constructions in \cref{sec:PCPP}.
\end{description}

\medskip
\parhead{Lower bounds:} For lower bounds, the only bound we are aware of is that of \cite{GurL19}, who proved that any $\Query$-query relaxed locally decodable code must have a block length $\BlockLength \geq \CodeDim^{1+\Omega(\frac{1}{\Query^2})}$.

For the strict notion of locally decodable codes, it is known by \cite{KatzTrevisan00, Woodruff07}
that for $\Query \geq 3$ any $\Query$-query LDC must have block length $\BlockLength \geq \Omega(\CodeDim^{1+\frac{1}{\ceil{\Query/2}-1}})$.
For $\Query=3$ a slightly stronger bound of $\BlockLength \geq \Omega(\CodeDim^2/\log(\CodeDim))$ is known,
and furthermore, for $3$-query \emph{linear} LDC the block length must be $\BlockLength \geq \Omega(\CodeDim^2/\log\log(\CodeDim))$~\cite{Woodruff07}.
For $\Query=2$ \cite{KdW03} proved an exponential lower bound of $\BlockLength \geq \exp(\Omega(\CodeDim))$.
See also \cite{DJKLR02, GKST02, Oba02, WdW05, Woodruff10} for more related work on lower bounds for LDCs.

\subsection{Organization}
The rest of the paper is organized as follows.
In \cref{sec:proof-overview}, we informally discuss the construction and the correcting algorithm.
In this discussion we focus on decoding the symbols corresponding to the message, i.e., on showing that the code is an RLDC.
\cref{sec:preliminaries} introduces the formal definitions and notations we will use in the proof of \cref{thm:main}.
We present the notion of consistency test using random walk in \cref{sec:CTRW-RM}, and prove that the Reed-Muller code admits such test.
In \cref{sec:PCPP} we present the PCPPs we will use in our construction, and state the properties needed for the correcting algorithm.
In \cref{sec:composition-local-alg} we prove \cref{thm:main}
by proving a composition theorem, which combines the instantiation of the Reed-Muller code with PCPPs from the previous sections.

\doclearpage
\section{Proof overview}
\label{sec:proof-overview}
In this section we informally describe our code construction.
Roughly speaking, our construction consists of two parts:
\begin{description}
  \item [The Reed-Muller encoding:] Given a message $\message \in \Field^\CodeDim$, its Reed-Muller encoding
  is the evaluation of an $\RMDim$-variate polynomial of degree at most $\RMDeg$ over $\Field$,
  whose coefficients are determined by the message we wish to encode.
  \item [Proofs of proximity:] The second part of the encoding consists of the concatenation of PCPPs,
  each claiming that a certain restriction of the first part agrees with some Reed-Muller codeword.
\end{description}
Specifically, given a message $\message \in \F^\CodeDim$, we first encode it using the Reed-Muller encoding $\RM_{\F}(\RMDim, \deg)$,
where $\RMDim$ roughly corresponds to the query complexity of our RLDC, and the field is large enough so that the distance of
the Reed-Muller code, which is equal to $1-\frac{\RMDeg}{\abs{\Field}}$, is some constant, say $3/4$.
That is, the first part of the encoding corresponds to an evaluation of some
polynomial $f \colon \Field^\RMDim \to \Field$ of degree at most $\RMDeg$.
The second part of the encoding consists of a sequence of PCPPs claiming that the restrictions of a the Reed-Muller part
to some carefully chosen planes in $\Field^\RMDim$ are evaluations of some low-degree polynomial.

The planes we choose are of the form $\Plane{\veca}{\vech}{\vech'} = \{\veca + t \cdot \vech + s \cdot \vech' : t,s \in \Field\}$,
where $\veca \in \Field^\RMDim$, and $\vech, \vech' \in \H^\RMDim$ for some $\H$ subfield of $\Field$.
We will call such planes \defemph{$\H$-planes}.
In order to obtain the RLDC with the desired parameters, we choose the field $\H$ so that $\Field$ is the extension of $\H$ of degree $[\Field:\H] = \RMDim$.
It will be convenient to think of $\H$ as a field and think of $\Field$ as a vector space of $\H$ of dimension $\RMDim$
(augmented with the multiplicative structure on $\Field$).
Indeed, the saving in the block length of the RLDC we obtain
crucially relies on the fact that we ask for PCPPs for only a small collection of planes,
and not all planes in $\Field^\RMDim$.
The actual constraints required to be certified by the PCPPs are slightly more complicated, and we describe the next.

The constraints of the first type correspond to $\H$-planes $\P$ and points $\vecx \in \P$.
For each such pair $(\P,\vecx)$ the code will contain a PCPP certifying that
\begin{inparaenum}[(i)]
\item the restriction of the Reed-Muller part to $\P$ is close to an evaluation of some polynomial of total degree at most $\RMDeg$,
\item and furthermore, this polynomial agrees with the value of the Reed-Muller part on $\vecx$.
\end{inparaenum}
In order to define it formally, we introduce the following notation.
\begin{notation}\label{notation:f_P^x}
    Let $\Field$ be a finite field of size $\RMLen$.
    Fix $\Word \colon \Field^\RMDim \to \Field$, a plane $\P$ in $\Field^\RMDim$, and a point $\vecx \in \P$.
    Denote $\Word_{\mid \P}^{(\vecx)} = \Word_{\mid \P} \circ (\Word(\vecx))^{\RMLen^2}$.
    That is, the length of $\Word_{\mid \P}^{(\vecx)}$ is $2 \cdot \RMLen^2$,
    and it consists of $\Word_{\mid \P}$ concatenated with $\RMLen^2$ repetitions of $\Word(\vecx)$.
\end{notation}
Given the notation above, if $\Word$ is the first part of the codeword, corresponding to the Reed-Muller encoding of the message,
then the PCPP for the pair $(\P,\vecx)$ is expected to be the proof of proximity
claiming that $\Word_{\mid \P}^{(\vecx)}$ is close to the language
\begin{equation}\label{def:RM_P^x}
   \RM_{\mid \P}^{(\vecx)} = \{Q \circ (Q(\vecx))^{(\RMLen^2)} : \mbox{$Q$ is the evaluation of a degree-$\RMDeg$ polynomial on $\P$} \} \seq \Field^{2\RMLen^2}
   \enspace.
\end{equation}
Note that by repeating the symbol $Q(\vecx)$ for $\RMLen^2$ times, the definition indeed puts weight 1/2 on the constraint that
the input $\Word_{\mid \P}$ is close to some low-degree polynomial $Q$, and puts weight 1/2 of the constraint $\Word(\vecx) = Q(\vecx)$.
In particular, if $\Word_{\mid \P}$ is $\delta$-close to some bivariate low degree polynomial $Q$
for some small $\delta>0$, but $\Word(\vecx) \neq Q(\vecx)$, then $\Word_{\mid \P}$ is at least $(1 - \frac{\RMDeg}{\abs{\Field}}-\delta)/2$-far from any bivariate low degree polynomial on $\P$.

The constraints of second type correspond to $\H$-planes $\P$ and lines $\ell \seq \P$.
For each such pair $(\P,\ell)$ the code will contain a PCPP certifying that
\begin{inparaenum}[(i)]
\item the restriction of the Reed-Muller part to $\P$ is close to an evaluation of some polynomial of total degree at most $\RMDeg$,
\item and furthermore, this polynomial is close to $\Word_{\mid \ell}$.
\end{inparaenum}
(In particular, this implies that $\Word_{\mid \ell}$ is close to some low-degree polynomial.)

\medskip
Next, we introduce the notation analogous to \cref{notation:f_P^x} replacing the points with lines.
\begin{notation}\label{notation:f_P^ell}
    Let $\Field$ be a finite field of size $\RMLen$.
    Fix $\Word \colon \Field^\RMDim \to \Field$, a plane $\P$ in $\Field^\RMDim$, and a line $\ell \seq \P$.
    Denote by $\Word_{\mid \P}^{(\ell)} = \Word_{\mid \P} \circ (\Word_{\mid \ell})^{\RMLen}$.
    That is, the length of $\Word_{\mid \P}^{(\ell)}$ is $2 \cdot \RMLen^2$,
    and it consists of $\Word_{\mid \P}$ concatenated with $\RMLen$ repetitions of $\Word_{\mid \ell}$.
\end{notation}
If $\Word$ is the Reed-Muller part of the codeword, corresponding to the Reed-Muller encoding of the message,
then the PCPP for the pair $(\P,\ell)$ is expected to be the proof of proximity
claiming that $\Word_{\mid \P}^{(\ell)}$ is close to the language
\begin{equation}\label{def:RM_P^ell}
   \RM_{\mid \P}^{(\ell)} = \{Q \circ (Q_{\mid \ell})^{\RMLen} :
            \mbox{$Q$ is the evaluation of some degree-$\RMDeg$ polynomial on $\P$} \} \seq \Field^{2\RMLen^2}
   \enspace.
\end{equation}
Again, similarly to the first part, repeating the evaluation of $Q_{\mid \ell}$ for $\RMLen$ times puts weight 1/2 on the constraint that
the input $\Word_{\mid \P}$ is a close to some low-degree polynomial $Q$, and puts weight 1/2 of the constraint $\Word_{\mid \ell}$ is close to $Q_{\mid \ell}$.

\medskip

With the proofs specified above, we now sketch the local correcting algorithm for the code.
Below we only focus on correcting symbols from the Reed-Muller part. Correcting the symbols from the PCPP part
follows a rather straightforward adaptation of the techniques from \cite{CGS20},
and we omit them from the overview.

Given a word $w \in \Field^\BlockLength$ and an index $i \in [\BlockLength]$ of $w$ corresponding to the Reed-Muller part
of the codeword,
let $\Word \colon \Field^\RMDim \to \Field$ be the Reed-Muller part of $w$,
and let $\vecx \in \Field^\RMDim$ be the input to $\Word$ corresponding to the index $i$.
The local decoder works in two steps.
\begin{description}
\item[Consistency test using random walk:] In the first step the correcting algorithm
    invokes a procedure we call \emph{consistency test using a random walk ($\CTRW$)} for the Reed-Muller code.
    This step creates a sequence of $\H$-planes of length $(\RMDim+1)$,
    where each plane defines a constraint checking that the restriction of $w$ to the plane is low-degree.
    Hence, we get $\RMDim+1$ constraints, each depending on $\RMLen^2$ symbols.

\item[Composition using proofs of proximity:]
    Then, instead of reading the entire plane for each constraint, we use the PCPPs from the second part of the
    codeword to reduce the arity of each constraint to $O(1)$,
    thus reducing the total query complexity of the correcting algorithm to $\Query = O(\RMDim)$.
    That is, for each constraint we invoke the corresponding PCPP verifier to check
    that the restrictions of $\Word$ to each of these planes is (close to) a low-degree polynomial.
    If at least one of the verifiers rejects, then the word $\Word$ must be corrupt, and hence the correcting algorithm returns $\bot$.
    Otherwise, if all the PCPP verifiers accept, the correcting algorithm returns $\Word(\vecx)$.
\end{description}

In particular, if $\Word$ is a correct Reed-Muller encoding, then the algorithm will always return $\Word(\vecx)$,
and the main part of the analysis is to show that if $\Word$ is close to some $Q^* \in \RM_{\F}(\RMDim, \deg)$,
but $\Word(\vecx) \neq Q^*(\vecx)$, then the correcting algorithm catches an inconsistency,
and returns $\bot$ with some constant probability. See \cref{sec:correction-alg-RM} for details.

\medskip

The key step in the analysis says that if $\Word$ is close to some codeword $Q^* \in \RM$ but $\Word(\vecx) \neq Q^*(\vecx)$,
then with high probability $\Word$ will be far from a low degree polynomial on at least one of these planes,
where ``far'' corresponds to the notion of distances defined by the languages $\RM_{\mid \P}^{(\vecx)}$ and $\RM_{\mid \P}^{(\ell)}$.
In particular, if on one of the planes $\Word$ is far from the corresponding language,
then the PCPP verifier will catch this with constant probability,
thus causing the correcting algorithm to return $\bot$.
We discuss this part in detail below.

\medskip

It is important to emphasize that the main focus of this work is constructing a correcting algorithm for the Reed-Muller part.
Using the techniques developed in \cite{CGS20}, it is rather straightforward to design the
algorithm for correcting symbols from the PCPPs part of the code.
See \cref{sec:correction-alg-PCPP} for details.

\subsection{$\CTRW$ on Reed-Muller codes}
\label{sec:overview-CTRW}
Below we define the notion of \emph{consistency test using random walk ($\CTRW$)} for the Reed-Muller code.
This notion is a slight modification of the notion originally defined in \cite{CGS20} for general codes.
In this paper we define it only for the Reed-Muller code.
Given a word $\Word \colon \Field^\RMDim \to \Field$ and some $\Index \in \Field^\RMDim $,
the goal of the test is to make sure that $\Word(\Index)$
is consistent with the codeword of Reed-Muller code closest to $\Word$.
\cite{CGS20} describe a $\CTRW$ for the tensor power $\Code^{\otimes m}$ of an arbitrary codes $\Code$ with good distance (e.g., Reed-Solomon).
The $\CTRW$ they describe works by starting from the point we wish to correct,
and choosing an axis-parallel line $\ell_1$ containing the starting point.
The test continues by choosing a sequence of random axis-parallel lines $\ell_2,\ell_3,\dots\ell_t$,
such that each $\ell_i$ intersects the previous one, $\ell_{i-1}$,
until reaching a uniformly random coordinate of the tensor code.
That is, the length of the sequence $t$ denotes the \emph{mixing time of the corresponding random walk}.
The predicates are defined in the natural way;
namely, the test expects to see a codeword of $\Code$ on each line it reads.

In this work, we present a $\CTRW$ for the Reed-Muller code,
which is a variant of the $\CTRW$ described above.
The main differences compared to the description above are that
\begin{inparaenum}[(i)]
    \item the test chooses a sequence of planes $\P_1,\P_3, \dots \P_t$ (and not lines),
    \item and every two planes intersect on a line (and not on a point).
\end{inparaenum}
Roughly speaking, the algorithm works as follows.

\begin{enumerate}
  \item Given a point $\Index \in \Field^\RMDim$ the test picks a uniformly random $\H$-plane $\P_0$ containing $\Index$.

  \item Given $\P_0$, the test chooses a random line $\ell_1 \seq \P_0$, and then chooses another random $\H$-plane $\P_1 \seq \Field^\RMDim$ containing $\ell_1$.

  \item Given $\P_1$, the test chooses a random line $\ell_2 \seq \P_1$, and then chooses another random $\H$-plane $\P_2 \seq \Field^\RMDim$ containing $\ell_2$.

  \item The algorithm continues for some predefined number of iterations, choosing $\P_0,\P_1,\P_2, \dots \P_\NumRep$.
        Roughly speaking, the number of iterations is equal to the mixing time of the corresponding Markov chain.
        More specifically, the process continues until a uniformly random point in $\P_\NumRep$ is close to a uniform point in $\Field^\RMDim$.

  \item The constraints defined for each $\P_i$ are the natural constraints;
        namely checking that the restriction of $\Word$ to $\P_i$ is a polynomial of degree at most $\RMDeg$.
\end{enumerate}

One of the important parameters, directly affecting the query complexity of our construction is the mixing time of the random walk.
Indeed, as explained above, the query complexity of our RLDC is proportional to the mixing time of the random walk.
We prove that if $[\Field:\H] = \RMDim$, then the mixing time is upper bounded by $\RMDim$.
In order to prove this we use the following claim, saying that if
$\F$ is the field extension of $\H$ of degree $\TensorDim$,
and $\vech_1,\dots,\vech_\TensorDim \in \H^\TensorDim$ and $t_1,\dots,t_\TensorDim \in \Field$ are
sampled uniformly, independently from each other,
then $\sum_{i=1}^\TensorDim t_i \cdot \vech_i$ is close to a uniformly random point in $\Field^\RMDim$.
See \cref{claim:mixing-time-key-claim} for the exact statement.

\medskip

As explained above, the key step of the analysis is to prove that if $\Word$ is close to some codeword $Q^* \in \RM$ but $\Word(\vecx) \neq Q^*(\vecx)$,
then with high probability at least one of the predicates defined will be violated.
Specifically, we prove that with high probability the violation will be in the following strong sense.

\begin{theorem}[informal, see \cref{thm:H-ctrw}]
\label{thm:H-ctrw-informal}
  If $\Word$ is close to some codeword $Q^* \in \RM$ but $\Word(\vecx) \neq Q^*(\vecx)$,
  then with high probability
  \begin{enumerate}
    \item either $\Word_{\mid \P_0}^{(\Index)}$ is $\Omega(1)$-far from $\RM_{\mid \P_0}^{(\Index)}$,
    \item or $\Word_{\mid \P_i}^{(\ell_i)}$ is $\Omega(1)$-far from $\RM_{\mid \P_i}^{(\ell_i)}$ for some $i \in [\RMDim]$.
  \end{enumerate}
\end{theorem}

Indeed, this strong notion of violation allows us to use the proofs of proximity in order
to reduce the query complexity to $O(1)$ queries for each $i \in [\RMDim]$.
We discuss proofs of proximity next.

\subsection{PCPs of proximity and composition}
The second building block we use in this work is the notion of \emph{probabilistic checkable proofs of proximity (PCPPs)}.
PCPPs were first introduced in \cite{BGHSV06} and \cite{DinurR04}.
Informally speaking, a PCPP verifier for a language $\Language$,
gets an oracle access to an input $x$ and a proof $\pi$ claiming that $x$ is close to some element of $\Language$.
The verifier queries $x$ and $\pi$ in some small number of (random) locations, and decides whether to accept or reject.
The completeness and soundness properties of a PCPP are as follows.
\begin{description}
  \item[Completeness:] If $x \in \Language$, then there exists a proof causing the verifier to accept with probability 1.
  \item[Soundness:] If $x$ is far from $\Language$, then no proof can make the verifier to accept with probability more than 1/2.
\end{description}
In fact, we will use the slightly stronger notion of \emph{canonical PCPP (cPCPP) systems}.
These are PCPP systems satisfying the following completeness and soundness properties.
For completeness, we demand that for each $w$ in the language there is a unique \emph{canonical} proof $\pi(w)$
that causes the verifier to accept with probability 1.
For soundness, the demand is that the only pairs $(x,\pi)$ that are accepted by the verifier with
high probability are those where $x$ is close to some $w \in \Language$ and $\pi$ is close to $\pi(w)$.
Such proof system have been studies in~\cite{DinurGG18, Paradise20},
who proved that such proof systems exist for every language in $\P$.

Furthermore, for our purposes we will demand a stronger notion of {correctable canonical PCPP systems (ccPCPP)}.
These are canonical PCPP systems where the set $\{w \circ \pi^*(w) : w \in \Language\}$ is a $\Query$-query RLCC for some parameter $\Query$,
with $\pi^*(w)$ denoting the canonical proof for $w \in \Language$.
It was shown in~\cite{CGS20} how to construct ccPCPP by combining a cPCPP system with \emph{any} systematic RLCC.
Informally speaking, for every $w \in \Language$, and its canonical proof $\pi(w)$,
we define $\pi^*(w)$ by encoding $w \circ \pi(w)$ using a systematic RLCC.
The verifier for the new proof system is defined in a straightforward manner.
See \cite{CGS20} for details.

\medskip

The PCPPs we use throughout this work, are the proofs of two types, certifying that
\begin{enumerate}
\item $\Word_{\mid \P}^{(\vecx)}$ is close to $\RM_{\mid \P}^{(\vecx)}$ for some plane $\P$ and some $\vecx \in \P$, and
\item $\Word_{\mid \P}^{(\ell)}$ is close to $\RM_{\mid \P}^{(\ell)}$ for some plane $\P$ and some line $\ell \seq \P$.
\end{enumerate}

Indeed, it is easy to see that the first type of proofs checks that
\begin{inparaenum}[(i)]
\item the restriction of $\Word$ to $\P$ is close to an evaluation of some polynomial $Q^*$ of total degree at most $\RMDeg$,
\item and $\Word(\vecx) = Q^*(\vecx)$.
\end{inparaenum}
Similarly, the second type proof certifies that
\begin{inparaenum}[(i)]
\item the restriction of $\Word$ to $\P$ is close to an evaluation of some polynomial $Q^*$ of total degree at most $\RMDeg$,
\item and $\Word_{\mid \ell}$ is close to $Q^*_{\mid \ell}$.
\end{inparaenum}

These notions of distance go together well with the guarantees we have for $\CTRW$ in \cref{thm:H-ctrw-informal}.
This allows us to \emph{compose} $\CTRW$ with the PCPPs to obtain a correcting algorithm with query complexity $\Query = O(\RMDim)$.
Informally speaking, the composition theorem works as follows.
We first run the $\CTRW$ to obtain a collection of $\RMDim+1$ constraints on the planes $\P_0,\P_1,\dots,\P_\RMDim$.
By \cref{thm:H-ctrw-informal}, we have the guarantee that with high probability
either $\Word_{\mid \P_0}^{(\Index)}$ is $\Omega(1)$-far from $\RM_{\mid \P_0}^{(\Index)}$,
or $\Word_{\mid \P_i}^{(\ell_i)}$ is $\Omega(1)$-far from $\RM_{\mid \P_i}^{(\ell_i)}$ for some $i \in [\RMDim]$.
Then, instead of actually reading the values of $\Word$ on all these planes,
we run the PCPP verifier on $\Word_{\mid \P_0}^{(\Index)}$ to check that it is close to $\RM_{\mid \P_0}^{(\Index)}$,
and running the PCPP verifier on each of the $\Word_{\mid \P_i}^{(\ell_i)}$ to check that they are close to $\RM_{\mid \P_i}^{(\ell_i)}$.
Each execution of the PCPP verifier makes $O(1)$ queries to $\Word$ and to the proof, and thus the total query complexity will be indeed $O(\RMDim)$.
As for soundness, if $\Word_{\mid \P_0}^{(\Index)}$ is $\Omega(1)$-far from $\RM_{\mid \P_0}^{(\Index)}$,
or $\Word_{\mid \P_i}^{(\ell_i)}$ is $\Omega(1)$-far from $\RM_{\mid \P_i}^{(\ell_i)}$ for some $i \in [\RMDim]$,
then the corresponding verifier will notice an inconsistency with constant probability, causing the decoder to output $\bot$.

We discuss proofs of proximity in \cref{sec:PCPP}.
The composition is discussed in \cref{sec:composition-local-alg}.

\doclearpage
\section{Preliminaries}
\label{sec:preliminaries}
We begin with standard notation.
The relative distance between two strings $x, y \in \Sigma^n$ is defined as
\begin{equation*}
    \dist(x,y) \DefineEqual \frac{\left| \left\{ i \in [n] : x_i \neq y_i \right\} \right|}{n}
    \enspace.
\end{equation*}
If $\dist(x,y) \leq \eps$, we say that $x$ is \emph{$\eps$-close} to $y$; otherwise we say that $x$ is \emph{$\eps$-far} from $y$.
For a non-empty set $S \subseteq \Sigma^n$ define the distance of $x$ from $S$ as $\dist(x,S) \DefineEqual \min_{y \in S} \dist(x,y)$.
If $\dist(x,S) \leq \eps$, we say that $x$ is \emph{$\eps$-close} to $S$; otherwise we say that $x$ is \emph{$\eps$-far} from $S$.

We will also need a more general notion of a distance, allowing different coordinates to have different weight.
In particular, we will need the distance that gives constant weight to a particular subset of the coordinates,
and spreads the rest of the weight uniformly between all coordinates.
\begin{definition}\label{def:weighted-distance}
Fix $n \in \N$ and an alphabet $\Sigma$.
For a set $A \seq [n]$ define the distance $\dist_{A}$ between two strings $x,y \in \Sigma^n$ as
\begin{equation*}
    \dist_{A}(x,y) = \frac{\abs{\{i \in A : x_i \neq y_i\}}}{2\abs{A}} + \frac{\abs{\{i \in [n] : x_i \neq y_i\}}}{2n}
    \enspace.
\end{equation*}
In particular, if $x$ differs from $y$ on $\delta \abs{A}$ coordinates in $A$,
then $\dist_{A}(x,y)$ is at least $\frac{\delta}{2} + \frac{\delta\abs{A}}{2n}$.

\medskip
\noindent
We define $\dist_A$ between a string $x \in \Sigma^n$ and a set $S \seq \Sigma^n$ as
\begin{equation*}
    \dist_A(x,S) = \min_{y \in S}\dist_A(x,y)
    \enspace.
\end{equation*}
\end{definition}

\begin{remark}
This definition generalizes the definition of \cite{CGS20} of $\dist_\coord$ for a coordinate $\coord \in [n]$.
Indeed, the notion of $\dist_\coord$ for a coordinate $\coord \in [n]$ corresponds to the singleton set $A = \{\coord\}$.
\end{remark}
When the set $A$ is a singleton $A = \{\coord\}$ we will write $\dist_\coord(x,y)$ to denote $\dist_{\{\coord\}}(x,y)$,
and we will write $\dist_\coord(x,S)$ to denote $\dist_{\{\coord\}}(x,S)$.



\subsection{Basic coding theory}
\label{sec:basic-coding-theory}

Let $k < n$ be positive integers, and let $\Sigma$ be an alphabet.
An \textdef{error correcting code} $\Code \colon \Sigma^k \to \Sigma^n$ is an \emph{injective} mapping from messages of length $k$ over the alphabet $\Sigma$ to codewords of length $n$.
The parameter $k$ is called the \emph{message length} of the code, and $n$ is its \emph{block length} (which we view as a function of $k$).
The \emph{rate} of the code is defined as $k/n$,
and the \emph{relative distance} of the code is defined as $\min_{\message \neq \message' \in \Sigma^k} \dist(\Code(\message),\Code(\message'))$.
We sometimes abuse the notation and use $\Code$ to denote the set of all of its codewords,
i.e., identify the code with $\{\Code(\message): \message \in \Sigma^k\} \subseteq \Sigma^n$.

\parhead{Linear codes}
Let $\Field$ be a finite field. A code $\Code \colon \Field^k \to \Field^n$ is \emph{linear} if it is an $\Field$-linear map from $\Field^k$ to $\Field^n$. In this case the set of codewords
$\Code$ is a subspace of $\Field^n$, and the message length of $\Code$ is also the dimension of the subspace.
It is a standard fact that for any linear code $\Code$,
the relative distance of $\Code$ is equal to $\min_{x \in \Code \setminus \{0^n\}} \dist(x,0^n)$.

\subsection{Reed-Muller codes}
\label{sec:rm-code}
Reed-Muller codes \cite{Muller1954} are among the most well studied error correcting codes, with many
theoretical and practical applications
in different areas of computer science and information theory.
Let $\F$ be a finite field of order $\abs{\Field} = \RMLen$, and let $\RMDeg$ and $\RMDim$ be integers.
The code $\RM_\Field(\RMDim, \RMDeg)$ is the linear code whose codewords
are the evaluations of polynomials $\Word \colon \Field^\RMDim \to \Field$ of total degree at most $\RMDeg$ over $\Field$.
We will allow ourselves to write $\RM(\RMDim, \RMDeg)$, since the field is fixed throughout the paper.
We will also sometimes omit the parameters $\RMDim$ and $\RMDeg$, and simply write $\RM$, when the parameters are clear from the context.

In this paper we consider the setting of parameters where $\RMDeg < \abs{\Field} = \RMLen$.
It is well known that for $\RMDeg < \RMLen$ the relative distance of $\RM_\Field(\RMDim, \RMDeg)$ is $1 - \frac{\RMDeg}{\RMLen}$.
The dimension of $\RM$ can be computed by counting the number of monomials of total degree at most $\RMDeg$.
For $\RMDeg < \RMLen$ the number of such monomials is $\binom{\RMDeg + \RMDim}{\RMDim} \geq (\frac{\RMDeg+\RMDim}{\RMDim})^{\RMDim} > (\frac{\RMDeg}{\RMDim})^{\RMDim}$.
Since the length of each codeword is $\RMLen^{\RMDim}$,
it follows that the rate of the code is $\frac{\binom{\RMDeg + \RMDim}{\RMDeg}}{\RMLen^{\RMDim}} > (\frac{\RMDeg}{\RMDim \RMLen})^{\RMDim}$.

\begin{definition}\label{def:line-and-plane}
For $\vecx, \vecy \in \F^\TensorDim$ denote by $\Line{\vecx}{\vecy}$ the line
\begin{equation*}
    \Line{\vecx}{\vecy} = \{\vecx + t \cdot \vecy : t \in \F\}
    \enspace.
\end{equation*}

Also, for $\vecx, \vecy, \vecz \in \F^\TensorDim$ denote by $\Plane{\vecx}{\vecy}{\vecz}$ the plane
\begin{equation*}
    \Plane{\vecx}{\vecy}{\vecz} = \{\vecx + t \cdot \vecy + s \cdot \vecz : t,s \in \F\}
    \enspace.
\end{equation*}
\end{definition}

An important property of $\RM(\RMDim, \RMDeg)$ (and multivariate low-degree polynomials, in general)
that we use throughout this work is that their restrictions to lines and planes in $\F^{\RMDim}$
are also polynomials of degree at most $\RMDeg$. In other words,
if $\Word \in \RM(\RMDim, \RMDeg)$, and $\ell$ is a line ($\P$ is a plane) in $\Field^\RMDim$,
then the restriction of $\Word$ to $\ell$ (or to $\P$) is a codeword of the Reed-Muller code of the same degree and lower dimension.

\medskip

The following lemma is a standard lemma in the PCP literature, saying that random lines sample well
the space $\Field^\RMDim$.

\begin{lemma}\label{lem:sampling}
Let $\F$ be a finite field.
For any subset $A \seq \F^2$ of density $\mu = |A|/|\F^2|$,
and for any $\epsilon >0$ it holds that
\begin{equation*}
    \Pr_{\vecx \in \F^2,\vecy \in \F^2}
    \bigg[\bigg|\frac{|\Line{\vecx}{\vecy}\cap A|}{|\Line{\vecx}{\vecy}|} - \mu \bigg| > \epsilon \bigg]
    \leq \frac{1}{|\F|} \cdot \frac{\mu}{\epsilon^2}
    \enspace.
\end{equation*}
\end{lemma}

\begin{proof}
    For each $t \in \F$, let $X_t$ be an indicator random variable for the event $\vecx+t \cdot \vecy \in A$.
    Since each point is a uniform point in the plane, we have $\E[X_t] = \Pr[X_t = 1]= \mu$,
    Therefore, denoting $X = \sum_{t \in \F}{X_t}$, it follows that $\E[\Line{\vecx}{\vecy}\cap A] = \E[X] = \mu \cdot \abs{\Field}$.

    We are interested in bounding the deviation of $X = \sum_{t}{X_t}$ from its expectation.
    We do it by bounding the variance of $X$.
    Note first that $\Var[X_t] = \mu - \mu^2 \leq \mu$.
    By the pairwise independence of the points on a line, it follows that
    $\Var[X] = \sum_{t \in \Field}\Var[X_t] \leq \mu \cdot \abs{\F}$.
    Therefore, by applying Chebyshev's inequality we get
    \begin{equation*}
        \Pr\left[ \bigg| \frac{|\Line{\vecx}{\vecy}\cap A|}{|\Line{\vecx}{\vecy}|} - \mu \bigg| > \epsilon \right]
        =
        \Pr\left[\abs{X-\mu \abs{\Field}} > \epsilon \abs{\F}\right]
        \leq \frac{\Var[X]}{(\epsilon\abs{\F})^2} \leq \frac{\mu}{\abs{\F}\cdot \epsilon^2}
        \enspace,
    \end{equation*}
    as required.
\end{proof}

The following claim will be an important step in our analysis.

\begin{claim}\label{claim:mixing-time-key-claim}
    Let $\RMDim \in \N$ be a parameter, let $\H$ be a finite field, and let $\F$ be its extension of degree $\RMDim$.
    Let $\vech_1,\dots,\vech_\RMDim \in \H^\RMDim$ and $t_1,\dots,t_\RMDim \in \F$ be
    chosen independently uniformly at random from their domains.

    Then for any set $A \seq \F^\RMDim$ of size $\abs{A} = \alpha \cdot \abs{\F^\RMDim}$ it holds that
    \begin{equation*}
        \Pr\left[\sum_{i=1}^\RMDim t_i \cdot \vech_i \in A \right] \leq \alpha + 2/\H
        \enspace.
    \end{equation*}
\end{claim}

\begin{proof}
In order to prove the claim let us introduce some notation.
We write each element in $\F$ as an $\RMDim$-dimensional row vector over $\H$.
Also, we will represent an element $\vecx \in \F^\RMDim$ as a $\RMDim \times \RMDim$ matrix over $\H$,
where the $i$'th row represents $\vecx_i \in \F$, the $i$'th coordinate of $\vecx$.
Using this notation we need to prove that the random matrix corresponding to the sum $\sum_{i=1}^\RMDim t_i \cdot \vech_i$
is close to a random matrix with entries chosen uniformly from $\H$ independently from each other.

Using the notation above, write each $t_i \in \F$ as a row vector $(t_{i,1},\dots,t_{i,\RMDim}) \in \H^\RMDim$.
Observe that for any vector $\vech_i = (\vech_{i,1},\dots,\vech_{i,\RMDim})^T \in \H^m$ we can represent $t_i \cdot \vech_i \in \F^\RMDim$ as the outer product
\begin{equation*}
    t_i \cdot \vech_i =
    \begin{bmatrix}
        t_{i,1} \cdot \vech_{i,1} & t_{i,2} \cdot \vech_{i,1}& \dots & t_{i,\RMDim} \cdot \vech_{i,1}\\
        t_{i,1} \cdot \vech_{i,2} & t_{i,2} \cdot \vech_{i,2}& \dots & t_{i,\RMDim} \cdot \vech_{i,2}\\
        \vdots & \vdots & \ddots &  \vdots\\
        t_{i,1} \cdot \vech_{i,\RMDim} & t_{i,2} \cdot \vech_{i,\RMDim}& \dots & t_{i,\RMDim} \cdot \vech_{i,\RMDim}\\
    \end{bmatrix}
    =
    \begin{bmatrix}
        \vech_{i,1} \\
        \vech_{i,2} \\
        \vdots \\
        \vech_{i,\RMDim}
    \end{bmatrix}
    \cdot
    \begin{bmatrix}
        t_{i,1} & t_{i,2} &\dots & t_{i,\RMDim}
    \end{bmatrix}
\end{equation*}
Therefore, the sum $\sum_{i=1}^\RMDim t_i \cdot \vech_i$ is represented as
\begin{equation*}
    \sum_{i=1}^\RMDim
    \begin{bmatrix}
        \vech_{i,1} \\
        \vech_{i,2} \\
        \vdots \\
        \vech_{i,\RMDim}
    \end{bmatrix}
    \cdot
    \begin{bmatrix}
        t_{i,1} & t_{i,2} &\dots & t_{i,\RMDim}
    \end{bmatrix}
    = H \cdot T
    \enspace,
\end{equation*}
where $H$ is the $\RMDim \times \RMDim$ matrix with $H_{i,j} = \vech_{j,i}$, and $T$ is the $\RMDim \times \RMDim$ matrix with $T_{i,j} = t_{i,j}$.
That is, the sum $\sum_{i=1}^\RMDim t_i \cdot \vech_i$ is represented as a product of two uniformly random matrices over $\SubF$.

Next we show that if $H,T \in \SubF^{\RMDim \times \RMDim}$ are chosen uniformly at random and independently,
then for any collection $A$ of matrices of size $\abs{A} = \alpha \cdot \abs{\H^{\RMDim^2}}$
it holds that $\Pr[H \cdot T \in A] \leq \alpha + 2/\H$.
Indeed,
\begin{equation*}
    \Pr[H \cdot T \in A] \leq \Pr[H \cdot T \in A | \mbox{$H$ is invertible}] +
    \Pr[\mbox{$H$ is not invertible}]
    \enspace.
\end{equation*}
If $H$ is invertible, then for a uniformly random $T \in \SubF^{\RMDim \times \RMDim}$ the probability that $H \cdot T \in A$ is exactly $\alpha$,
and it is easy to check that $\Pr[\mbox{$H$ is not invertible}] = \sum_{i=1}^{\RMDim}\frac{1}{\abs{\H}^{i}} \leq \frac{2}{\abs{\H}}$.
\end{proof}

\subsection{Relaxed locally correctable codes}
\label{sec:rlcc}
Following the discussion in the introduction, we provide a formal definition of relaxed LCCs,
and state some related basic facts and known results.

\begin{definition}[Relaxed LCC]
\label{def:rlcc}
Let $\Code \colon \Sigma^\CodeDim \to \Sigma^\BlockLength$ be an error correcting code with relative distance $\DistCode$,
and let $\Query \in \N$, $\DecRad \in (0,\delta/2)$,and $\Soundness \in (0,1]$ be parameters.
Let $\Decoder$ be a randomized algorithm that gets an oracle access to an input $w \in \Sigma^n$ and an explicit access to an index $i \in [n]$.
We say that $\Decoder$ is a \defemph{$\Query$-query relaxed local correction algorithm} for $\Code$ with correction radius $\DecRad$ and soundness $\Soundness$
if for all inputs the algorithm $\Decoder$ reads explicitly the coordinate $i \in [\BlockLength]$, reads at most $\Query$ (random) coordinates in $w$, and satisfies the following conditions.
\begin{enumerate}
  \item \label{def:rlcc-decodes-correct}
  For every $w \in \Code$, and every coordinate $i \in [\BlockLength]$ it holds that $\Pr[\Decoder^w(i) = w_i]=1$.

  \item \label{def:rlcc-second-cond}
  For every $w \in \Sigma^n$ that is $\DecRad$-close to some codeword $c^* \in \Code$ and every coordinate $i \in [\BlockLength]$
  it holds that $\Pr[ \Decoder^w(i) \in \{ c^*_i,\bot \} ] \geq \Soundness$,
  where $\bot \not \in \Sigma$ is a special \emph{abort} symbol.
\end{enumerate}
The code $\Code$ is said to be a \defemph{$(\DecRad,\Soundness)$-relaxed locally correctable code (RLCC)} with query complexity $\Query$ if it admits a $\Query$-query relaxed local correction algorithm with correction radius $\DecRad$ and soundness $\Soundness$.
\end{definition}

%

\begin{observation}
    Note that for \emph{systematic} codes it is clear from \cref{def:rlcc} that RLCC is a stronger notion than RLDC,
    as it allows the local correction algorithm not only to decode each symbol of the message, but also each symbol of the codeword itself.
    That is, any systematic RLCC is also an RLDC with the same parameters.
\end{observation}

Finally, we recall the following theorem of Chiesa, Gur, and Shinkar~\cite{CGS20}.

\begin{theorem}[\cite{CGS20}]
\label{thm:CGS_RLCC}
For any finite field $\Field$, and parameters $\CodeDim, q  \in \N$,
there exists an explicit construction of a systematic linear code
$\Code_{\rm CGS} \colon \Field^\CodeDim \to \Field^\BlockLength$
with block length $\BlockLength = q^{O(\sqrt{q})} \cdot \CodeDim^{1 + O(1/\sqrt{q})}$
and constant relative distance,
that is a $q$-query RLCC with constant correction radius $\DecRad = \Omega(1)$, and constant soundness $\Soundness = \Omega(1)$.
\end{theorem}

\subsection{Canonical PCPs of proximity}
\label{sec:canonical-pcpps}

Next we define the notions of \emph{probabilistically checkable proofs of proximity},
and the variants that we will need in this paper.

\begin{definition}[PCP of proximity]
\label{def:PCPP}
A $\Query$-query \defemph{PCP of proximity (PCPP) verifier} for a language $\Language \subseteq \Sigma^*$
with soundness $\SoundnessPCPP$ with respect the to proximity parameter $\rho$, is a polynomial-time randomized algorithm $V$
that receives oracle access to an input $x \in \Sigma^n$ and a proof $\pi$.
The verifier makes at most $\Query$ queries to $x \circ \pi$ and has the following properties:
\begin{description}
  \item[Completeness:] For every $x \in \Language$ there exists a proof $\pi$ such that $\Pr[V^{x,\pi}=ACCEPT]=1$.
  \item[Soundness:] If $x$ is $\rho$-far from $\Language$, then for every proof $\pi$ it holds that $\Pr[V^{x,\pi}=ACCEPT] \leq \SoundnessPCPP$.
\end{description}
\end{definition}

A \emph{canonical PCPP (cPCPP)} is a PCPP in which every instance in the language has a \emph{canonical} accepting proof.
Formally, a canonical PCPP is defined as follows.

\begin{definition}[Canonical PCPP]
\label{def:canonical-PCPP}
A $\Query$-query \defemph{canonical PCPP verifier} for a language $\Language \subseteq \Sigma^{*}$ with soundness $\SoundnessPCPP$ with respect to proximity parameter $\rho$,
is a polynomial-time randomized algorithm $V$
that gets oracle access to an input $x \in \Sigma^n$ and a proof $\pi$.
The verifier makes at most $\Query$ queries to $x \circ \pi$, and satisfies the following conditions:
\begin{description}
  \item[Canonical completeness:]
  For every $w \in L$ there exists a \emph{unique} (canonical) proof $\pi(w)$
  for which $\Pr[V^{w,\pi(w)} = ACCEPT]=1$.

  \item[Canonical soundness:]
  For every $x \in\Sigma^n$ and proof $\pi$ such that
\begin{equation}
\label{strong-pcpp:eq}
\delta(x,\pi) \triangleq
\min_{w \in \Language}
\left\{
  \max\left(
    \frac{\dist(x,w)}{n}
    \ ,\
    \frac{\dist(\pi,\pi(w))} {\PCPPLen(n)}
  \right)
\right\} > \rho
\enspace,
\end{equation}
it holds that $\Pr[V^{x,\pi} = ACCEPT] \leq \SoundnessPCPP$.
\end{description}
\end{definition}

The following result on canonical PCPPs was proved in~\cite{DinurGG18} and \cite{Paradise20}.

\begin{theorem}[\cite{DinurGG18,Paradise20}]
\label{thm:PCPP}
Let $\Robust>0$ be a proximity parameter.
For every language in $\Language \in \PClass$ there exists a polynomial $\PCPPLen \colon \N \to \N$ and a canonical PCPP verifier for $\Language$ satisfying the following properties.
    \begin{enumerate}
      \item For all $x \in \Language$ of length $\abs{x} = n$ the length of the canonical proof $\pi(x)$ is $\abs{\pi(x)} = \PCPPLen(n)$.
      \item The query complexity of the PCPP verifier is $\Query = O(1/\Robust)$.
      \item The PCPP verifier for $\Language$ has perfect completeness and soundness $\Soundness = 1/2$ for proximity parameter $\Robust$ (with respect to the uniform distance measure).
    \end{enumerate}
\end{theorem}

Next, we define the stronger notion of \emph{correctable canonical PCPPs (ccPCPP)}, originally defined in \cite{CGS20}.
A ccPCPP system is a canonical PCPP system that in addition to allowing the verifier to be able to locally verify the validity of the given proof,
it also admits a local correction algorithm that locally corrects potentially corrupted symbols of the canonical proof.
Formally, the ccPCPP is defined as follows.
\begin{definition}[Correctable canonical PCPP]
\label{def:cPCP}
    A language $\Language \subseteq \Sigma^{*}$ is said to admit a ccPCPP with query complexity $\Query$ and soundness $\SoundnessPCPP$ with respect to the proximity parameter $\Robust$,
    and correcting soundness $\Soundness$ for correcting radius $\DecRad$ if it satisfies the following conditions.

\begin{enumerate}
    \item $\Language$ admits a $\Query$-query canonical PCPP verifier for $\Language$ satisfying the conditions in \cref{def:canonical-PCPP}
    with soundness $\SoundnessPCPP$ with respect to the proximity parameter $\Robust$.

    \item The code $\Pi_{\Language} = \{w \circ \pi(w) : w \in \Language\}$ is a $(\DecRad,\Soundness)$-RLCC with query complexity $\Query$,
    where $\pi(w)$ is the canonical proof for $w \in \Language$ from \cref{def:canonical-PCPP}.
\end{enumerate}
\end{definition}

\doclearpage
\section{Consistency test using random walk on the Reed-Muller code}
\label{sec:CTRW-RM}
Below we define the notion of \emph{consistency test using random walk ($\CTRW$)}.
This notion has been originally defined in \cite{CGS20} for tensor powers of general codes.
In this paper we focus on $\CTRW$ for the Reed-Muller code.

Informally speaking, a consistency test using random walk for Reed-Muller code $\RM = \RM_{\Field}(\RMDim, \deg)$
is a randomized algorithm that gets a word $\Word \colon \F^\RMDim \to \F$, which is close to some codeword $\codeword^* \in \RM$,
and an index $\Index \in \Field^\RMDim$ as an input, and its goal is to check whether $\Word(\Index) = \codeword^*(\Index)$.
In other words, it checks whether the value of $\Word$ at $\Index$ is consistent with the close codeword $\codeword^*$.
Below we formally describe the random process.

\begin{definition}[Consistency test using $\H$-plane-line random walk on $\RM_{\Field}(\RMDim, \deg)$]
\label{def:consistency-test}
    Let $\H$ be a field, and let $\F$ be a field extension of $\H$.
    Let $\RM = \RM_{\F}(\RMDim, \RMDeg)$ be the Reed-Muller code.
    An $\RWSteps$-steps \defemph{consistency test using $\SubF$-plane-line random walk} on $\RM$
    is a randomized algorithm that gets as input the evaluation table
    of some $\Word \colon \F^\RMDim \to \F$ and a coordinate $\Index \in \Field^\RMDim$,
    and works as in \cref{alg:H-CTRW}.

\begin{algorithm}
\caption{$\H$-plane-line $\CTRW$ for the $\RMDim-$dimensional Reed-Muller code}\label{alg:H-CTRW}

    \KwIn{$\Word \colon \F^{\TensorDim} \to \F$, $\Index \in \F^\TensorDim$}

    Pick $\vech_0, \vech'_0 \in \H^\TensorDim$ uniformly at random, and let $\vecx_0 = \vecx$

    Let $\P_0 = \Plane{\vecx_0}{\vech_0}{\vech'_0}$ be a random $\H$-plane passing through $\Index$

    \For{$i=1$ \KwTo $\RWSteps$}{
        Sample $s_{i-1},s'_{i-1} \in \F$ uniformly and independently

        Let $\vecx_i = \vecx_{i-1} + s_{i-1} \cdot \vech_{i-1} + s'_{i-1} \cdot \vech'_{i-1}$
        be a uniformly random point in $\P_{i-1}$

        Sample $t_{i-1}, t'_{i-1} \in \F$ uniformly and independently, and let $\vech_i = t_{i-1} \cdot \vech_{i-1} + t'_{i-1}\cdot \vech'_{i-1}$

        Let $\ell_i = \Line{\vecx_i}{\vech_i} = \{\vecx_i + t \cdot \vech_i : t \in \Field\}$ be a random line in $\P_{i-1}$

        Pick $\vech'_i \in \H^m$ uniformly at random

        Let $\P_i = \Plane{\vecx_i}{\vech_i}{\vech_i'}$

    }

    \uIf {$\Word_{\mid \P_i}$ is an evaluation of a polynomial of total degree at most $\RMDeg$ for all $0 \leq i \leq \RWSteps$}
    {
        \Return{ACCEPT}
    }

    \Else
    {
        \Return{REJECT}
    }
\end{algorithm}

    \medskip
    \noindent
    We say that $\CTRW$ has perfect completeness and $(\TestRadius, \Robust, \Soundness)$-robust soundness if it satisfies the following guarantees.
    \begin{description}[nolistsep]
        \item [Perfect completeness:] If $\Word \in \RM$, then $\Pr[\CTRW^{\Word}(\Index) = ACCEPT] = 1$ for all $\Index \in \F^\RMDim$.
        \item [$(\TestRadius, \Robust, \Soundness)$-robust soundness:]
        If $\Word$ is $\TestRadius$-close to some $\codeword^* \in \RM$, but $\Word(\Index) \neq \codeword^*(\Index)$,
        then
        \begin{equation*}
            \Pr[\textrm{$\dist_{\Index}(\Word_{\mid \P_0}, \RM_{\mid \P_0}) \geq \Robust
            \vee
            \exists i \in [\RWSteps]$ such that $\dist_{\ell_i}(\Word_{\mid \P_i}, \RM_{\mid \P_i}) \geq \Robust$}]
            \geq \Soundness
            \enspace.
        \end{equation*}
    \end{description}
    Here $\dist_{\Index}$ and $\dist_{\ell_i}$ are as in \cref{def:weighted-distance}.
    \end{definition}

    \begin{remark}
    Note that the soundness condition above is equivalent to checking that
    \begin{equation*}
        \Pr[\textrm{$\dist(\Word_{\mid \P_0}^{(\Index)}, \RM_{\mid \P_0}^{(\Index)}) \geq \Robust
            \vee
            \exists i \in [\RWSteps]$ such that $\dist(\Word_{\mid \P_i}^{(\ell_i)}, \RM_{\mid \P_i}^{(\ell_i)}) \geq \Robust$}]
        \geq \Soundness
        \enspace.
        \end{equation*}
    \end{remark}
\noindent
Next, we show that the Reed-Muller code admits an $\RMDim$-steps consistency test using $\SubF$-plane-line random walk
with constant robust soundness.

\begin{theorem}\label{thm:H-ctrw}
    For integer parameters $\RMDeg,\RMDim \geq 2$,
    let $\H$ be a prime field, and let $\Field$ be field extension of $\H$ of degree $[\Field:\H] = \RMDim$ such that  $\abs{\Field} \geq 2\TensorDim \RSdeg$.
    Denote the size of $\Field$ by $n = \abs{\Field}$.
    Let $\RM = \RM_{\F}(\RMDim, \RMDeg)$ be the Reed-Muller code over the field $\Field$,
    so that the distance of the code is $\RMDist \geq 1- 1/2\TensorDim \geq 3/4$.
    Then, for any $\TestRadius \leq \RMDist/2$ and $\Robust \leq \RMDist/8$
    the $\RMDim$-steps consistency test using $\SubF$-plane-line random walk on $\RM$
    has perfect completeness and $(\TestRadius, \Robust, \Soundness)-$robust soundness,
    with $\Soundness = \left(1 - \frac{4}{|\F|} \right)^{\RMDim} - \frac{\TestRadius + \frac{2}{\abs{\SubF}}}{\RMDist-2\Robust}$.
\end{theorem}

\begin{proof}
Consider an $\RWSteps$-steps consistency test using random walk on $\RM$ as in \cref{alg:H-CTRW}.
By construction, it is clear that whenever $\Word \in \RM$, the algorithm accepts.
It remains to prove the robust soundness of the algorithm.
Assume that $\Word$ is $\TestRadius-$close to some $\codeword^* \in \RM$. Note that since $\RM$ is a linear code, without loss of generality, we may assume that $\codeword^*$ is all-zeros codeword. Indeed, if $\Word$ is $\TestRadius$-close to some non-zero codeword $\codeword^*$, then we can consider the word $\Word' = \Word - \codeword^*$, which is $\TestRadius$-close to the all-zeros codeword, and behavior of the algorithm on both of these cases are the same. Hence, from now on we will assume that $\Word(\Index) \neq 0$, and $\Word$ is $\TestRadius-$close to the all-zeros codeword.
Below, we show that when running \cref{alg:H-CTRW} on such $\Word$, then for any choice of $\P_0$ we have either
\begin{equation}\label{eq:dist-P-0}
    \dist_{\Index}(\Word_{\mid \P_0}, \RM_{\mid \P_0}) \geq \Robust
\end{equation}
or
\begin{equation}\label{eq:H-ctrw-key-ineq}
    \Pr[\exists i \in [\RWSteps] \mathrm{\ s.t. \ } \mathrm{dist}_{\ell_i}(\Word_{\mid \P_i}, \RM_{\mid \P_i}) \geq \Robust] \geq \bigg(1 - \frac{4}{|\F|}\bigg)^{\RWSteps} - \frac{\TestRadius + \frac{2}{\abs{\SubF}}}{\RMDist-2\Robust}
    \enspace.
\end{equation}

\noindent
It is clear that each of \cref{eq:dist-P-0} and \cref{eq:H-ctrw-key-ineq} proves \cref{thm:H-ctrw}.

Clearly, if $\dist_{\Index}(\Word_{\mid \P_0}, \RM_{\mid \P_0}) \geq \Robust$, then we are done.
Hence, let us assume that $\dist_{\Index}(\Word_{\mid \P_0}, \RM_{\mid \P_0}) < \Robust$.
In particular, since $\Word(\Index) \neq 0$, $\Robust \leq \RMDist/8$, and $\dist_{\Index}(\Word_{\mid \P_0}, \RM_{\mid \P_0}) < \Robust$,
it follows that $\Word_{\mid \P_0}$ is $2\Robust$-close to some \emph{non-zero} codeword of $\RM_{\mid \P_0}$,
and hence $\Word_{\mid \P_0}$ contains at least $(\RMDist - 2\Robust)n^2$ non-zero entries.
For the rest of the proof we focus on proving \cref{eq:H-ctrw-key-ineq} assuming that $\Word_{\mid \P_0}$ contains at least $(\RMDist - 2\Robust)n^2$ non-zero entries.

In order to prove it, we introduce the events $E_i$ and $F_i$.

\begin{definition}
For $i \in [\RWSteps]$ denote by $E_i$ the event that $\Word_{\mid \ell_i}$ has at least $2\Robust \RSlen$ non-zeros, and $\Word_{\mid \P_i}$ has less than $(\RMDist-2\Robust)\RSlen^2$ non-zeros.
For $i \in [\RWSteps]$ denote by $F_i$ the event that $\Word_{\mid \ell_i}$ has at least $2 \Robust \RSlen$ non-zeros, and $\Word_{\mid \P_i}$ has at least $(\RMDist-2\Robust)\RSlen^2$ non-zeros.
\end{definition}

The following are the key observations about the event $E_i$
\begin{observation}
    If $E_i$ holds and $\Robust \leq \RMDist/4$, then
    \begin{enumerate}
      \item $\dist_{\ell_i}(\Word_{\mid \P_i}, \mathbf{0}) \geq \Robust$, since $\Word_{\mid \ell_i}$ has at least $2\Robust \RSlen$ non-zeros.
      \item $\dist_{\ell_i}(\Word_{\mid \P_i}, Q) \geq \Robust$ for all $\codeword \in \RM \setminus \{0\}$, since $\Word_{\mid \P_i}$ has less than $2\Robust \RSlen^2$ non-zeros.
    \end{enumerate}
    \noindent
    In particular, if $E_i$ holds, then $\dist_{\ell_i}(\Word_{\mid \P_i}, \RM_{\mid \P_i}) \geq \Robust$.
\end{observation}

For each $i \in [\RWSteps]$ denote
\begin{equation*}
    \epsilon_i = \Pr[(\wedge_{j=0}^{i-1} F_j) \bigwedge E_i ]
    \enspace.
\end{equation*}
Observe that the events corresponding to $\epsilon_i$'s are disjoint, and hence
\begin{equation*}\label{eqn:sumepsilons}
    \Pr[\exists i \in [\RWSteps] \mathrm{\ s.t. \ } \mathrm{dist}_{\ell_i}(\Word_{\mid \P_i}, \RM_{\mid \P_i}) \geq \Robust]
    \geq \sum_{i=1}^{\RWSteps}{\epsilon_i}
    \enspace.
\end{equation*}
The following two lemmas are the key steps in the proof of \cref{thm:H-ctrw}.
\begin{lemma}\label{lemma:random-z-in-Pm}
  For a uniformly random point $\vecz \in \P_\RMDim$ we have $\Pr[\Word(\vecz) \neq 0] \leq \TestRadius + \frac{2}{\abs{\SubF}}$.
\end{lemma}

\begin{lemma}\label{lemma:all-F-i-s-hold-whp}
    If $\dist_{\Index}(\Word_{\mid \P_0}, \RM_{\mid \P_0}) < \Robust$,
    then $\Pr[(\wedge_{i=1}^{\RWSteps} F_i)] > (1 - \frac{4}{|\F|})^{\RWSteps} - \sum_{i=1}^{\RWSteps} \epsilon_i$
    for all $\RWSteps \geq 1$.
\end{lemma}
We postpone the proofs of the lemmas for now,
and proceed with the proof of \cref{thm:H-ctrw} assuming the lemmas.

Note that if we choose a uniformly random $\vecz \in \P_\RMDim$, then
\begin{align*}
    \Pr[\Word(\vecz) \neq 0]
    \geq \Pr[(\wedge_{i=1}^{\RMDim} F_i) \wedge  \Word(\vecz) \neq 0]
    &= \Pr[(\wedge_{i=1}^{\RMDim} F_i)] \cdot \Pr[\Word(\vecz) \neq 0 \mid \wedge_{i=1}^{\RMDim} F_i] \\
    &\geq \Pr[(\wedge_{i=1}^{\RMDim} F_i)] \cdot (\RMDist-2\Robust)
    \enspace,
\end{align*}
where the last inequality is by noting that if we condition on $\wedge_{i=1}^{\RMDim} F_i$, then $\Word_{\mid \P_\RMDim}$ has at least $(\RMDist-2\Robust)\RSlen^2$ non-zeros, and hence $\Pr[\Word(\vecz) \neq 0 \mid \wedge_{i=1}^{\RMDim} F_i] \geq (\RMDist-2\Robust)$.
Therefore, by \cref{lemma:random-z-in-Pm} and \cref{lemma:all-F-i-s-hold-whp} it follows that
\begin{equation}\label{eqn:distance}
  \TestRadius + \frac{2}{\abs{\SubF}}
  \geq
  \Pr[\Word(\vecz) \neq 0]
  \geq
  \Pr[(\wedge_{i=1}^{\RMDim} F_i)] \cdot (\RMDist-2\Robust)
  \geq
  \left(\left(1 - \frac{4}{|\F|}\right)^{\RMDim} - \sum_{i=1}^{\RMDim} \epsilon_i \right) \cdot (\RMDist-2\Robust)
  \enspace,
\end{equation}
and hence
\begin{equation*}
    \Pr[\exists i \in [\RMDim] \mathrm{\ s.t. \ } \mathrm{dist}_{\ell_i}(\Word_{\mid \P_i}, \RM_{\mid \P_i}) \geq \Robust]
    \geq \sum_{i=1}^{\RMDim}{\epsilon_i}
    \geq \left(1 - \frac{4}{|\F|} \right)^{\RMDim} - \frac{\TestRadius + \frac{2}{\abs{\SubF}}}{\RMDist-2\Robust}
    \enspace.
\end{equation*}
This completes the proof of \cref{thm:H-ctrw}.
\end{proof}

We now return to the proof of \cref{lemma:random-z-in-Pm}.

\begin{proof}[Proof of \cref{lemma:random-z-in-Pm}]
Fix $\vecx_0$ in \cref{alg:H-CTRW}, and consider the independent choices of
$\{s_i,s'_i \in \Field\}_{i=0}^{\RMDim}$,
$\{t_i,t'_i \in \F\}_{i=0}^{\RMDim}$,
and $\{\vech'_i \in \SubF^\RMDim\}_{i=1}^{\RMDim}$.

Note first that for all $i \in [\RMDim]$ we have
\begin{equation*}
    \vech_i = (\prod_{j=0}^{i-1} t_j) \cdot \vech_0 + \sum_{u=0}^{i-1} (t'_u \cdot  \prod_{j=u+1}^{i-1} t_j) \cdot \vech'_u
    \enspace.
\end{equation*}
We prove this by induction. Indeed, by \cref{alg:H-CTRW}, we have $\vech_1 = t_0 \cdot \vech_0 + t'_0 \cdot \vech'_0$.
For the induction step, assume that the equation holds for some $i \in [\RMDim - 1]$.
Then for $i + 1$ we have
\begin{align*}
    \vech_{i+1} &= t_i \cdot \vech_i + t'_i \cdot \vech'_i \\
    &= t_i \cdot \left((\prod_{j=0}^{i-1} t_j) \cdot \vech_0 + \sum_{u=0}^{i-1} (t'_u \cdot  \prod_{j=u+1}^{i-1} t_j) \cdot \vech'_u \right) + t'_i \cdot \vech'_i \\
    &= (\prod_{j=0}^{i} t_j) \cdot \vech_0 + t_i \left(\sum_{u=0}^{i-1} (t'_u \cdot  \prod_{j=u+1}^{i-1} t_j) \cdot \vech'_u \right) + t'_i \cdot \vech'_i \\
    &= (\prod_{j=0}^{i} t_j) \cdot \vech_0 + \left(\sum_{u=0}^{i-1} (t'_u \cdot  \prod_{j=u+1}^{i} t_j) \cdot \vech'_u \right) + t'_i \cdot \vech'_i \\
    &= (\prod_{j=0}^{i} t_j) \cdot \vech_0 + \sum_{u=0}^{i} (t'_u \cdot  \prod_{j=u+1}^{i} t_j) \cdot \vech'_u
    \enspace,
\end{align*}
which concludes the induction step. Note that the first equation comes from the definition in \cref{alg:H-CTRW} and second equation follows from the induction hypothesis.
\par
Also, for all $i \in [\RMDim]$ we have
\begin{equation*}
    \vecx_i = \vecx_0 + \sum_{j=0}^{i-1} s_j \vech_j + \sum_{j=0}^{i-1} s'_j \vech'_j
    \enspace.
\end{equation*}
Again, we prove this by induction.
By \cref{alg:H-CTRW}, we have $\vecx_1 = \vecx_0 + s_0 \cdot \vech_0 + s'_0 \cdot \vech'_0$.
For the induction step, if we assume that the equation holds for some $i \in [\RMDim - 1]$,
then for $i+1$ we have
\begin{align*}
    \vecx_{i+1} &= \vecx_i + s_i \cdot \vech_i + s'_i \cdot \vech'_i \\
    &= \left(\vecx_0 + \sum_{j=0}^{i-1} s_j \vech_j + \sum_{j=0}^{i-1} s'_j \vech'_j \right) + s_i \cdot \vech_i + s'_i \cdot \vech'_i \\
    &= \vecx_0 + \sum_{j=0}^{i} s_j \vech_j + \sum_{j=0}^{i} s'_j \vech'_j
\end{align*}
This, completes the induction step. Note that the first equation comes from the definition in \cref{alg:H-CTRW} and second equation follows from the induction hypothesis. \par

Let $\P_\RMDim = \Plane{\vecx_\RMDim}{\vech_\RMDim}{\vech'_\RMDim}$
Note that we can sample $\vecz \in \P_\RMDim$ uniformly by choosing $s_{\RMDim}, s'_{\RMDim} \in \Field$,
and letting $\vecz = \vecx_\RMDim + s_{\RMDim} \cdot \vech_\RMDim + s'_{\RMDim} \vech'_{\RMDim}$.
Therefore,
\begin{align*}
    \vecz
    &= \left( \vecx_0 + \sum_{i=0}^{\RMDim} s_i \vech_i \right) + \left( \sum_{r=0}^{\RMDim} s'_r \vech'_r \right) \\
    &= \left( \vecx_0 + s_0 \vech_0 + \sum_{i=1}^{\RMDim} s_i \left( (\prod_{j=0}^{i-1} t_j) \cdot \vech_0
        + \sum_{r=0}^{i-1} (t'_r \cdot  \prod_{j=r+1}^{i-1} t_j) \cdot \vech'_r \right) \right) + \left( \sum_{r=0}^{\RMDim} s'_r \vech'_r \right) \\
    &= \left( \vecx_0 + (s_0 + \sum_{i=1}^{\RMDim} s_i (\prod_{j=0}^{i-1} t_j)) \cdot \vech_0 \right)
        + \sum_{i=1}^{\RMDim} s_i \left( \sum_{r=0}^{i-1} (t'_r \cdot  \prod_{j=r+1}^{i-1} t_j) \cdot \vech'_r \right)
        + \left( \sum_{r=0}^{\RMDim} s'_r \vech'_r \right) \\
    &= \left( \vecx_0 + (s_0 + \sum_{i=1}^{\RMDim} s_i (\prod_{j=0}^{i-1} t_j)) \cdot \vech_0 \right)
        + \sum_{r=0}^{\RMDim} \left( \sum_{i=r+1}^{\RMDim} (s_i t'_r \cdot  \prod_{j=r+1}^{i-1} t_j) +  s'_r \right) \vech'_r
    \enspace.
\end{align*}
Next, we fix all random choices except for $\{s'_r\}$ and $\{\vech'_r\}$,
and apply \cref{claim:mixing-time-key-claim}. Let $A$ be the set of indices $\Index$ such that $\Word(\Index) \neq 0$. Since $\Word$ is $\TestRadius$-close to all-zeros codeword, it immediately follows that $\abs{A} = \TestRadius \cdot \abs{\F^{\RMDim}}$. Since each $s'_r$ is chosen uniformly at random from $\F$, and each $\vech'_r$ is chosen uniformly at random from $\SubF^\RMDim$, by applying \cref{claim:mixing-time-key-claim} with respect to $A$, we have
\begin{align*}
  \Pr[\Word(\vecz) \neq 0]
        &= \Pr\left[\left( \vecx_0 + (s_0 + \sum_{i=1}^{\RMDim} s_i (\prod_{j=0}^{i-1} t_j)) \cdot \vech_0 \right)
        + \sum_{r=0}^{\RMDim} \left( \sum_{i=r+1}^{\RMDim} (s_i t'_r \cdot  \prod_{j=r+1}^{i-1} t_j) +  s'_r \right) \vech'_r \in A\right] \\
        &\leq \TestRadius + \frac{2}{\SubF} \enspace,
\end{align*}
which completes the proof of \cref{lemma:random-z-in-Pm}.
\end{proof}

Next we prove \cref{lemma:all-F-i-s-hold-whp}.

\begin{proof}[Proof of \cref{lemma:all-F-i-s-hold-whp}]
We lower-bound the value of $\Pr[(\wedge_{i=1}^{\RWSteps} F_i)]$ by peeling off one $F_i$ at a time.
Observe that for every $i \in [\RWSteps]$ we have
\begin{align}
\Pr[(\wedge_{j=1}^{i-1} F_j) \wedge  \mbox{$\ell_i$ has at least $2\Robust \RSlen$ non-zeros}]
&= \Pr[(\wedge_{j=1}^{i-1} F_j) \wedge F_i] + \Pr[(\wedge_{j=1}^{i-1} F_j) \wedge E_i] \nonumber \\
&= \Pr[(\wedge_{j=1}^{i} F_j)] + \epsilon_i \enspace. \label{eq:equality}
\end{align}

We will use the following claim.

\begin{claim}\label{claim:sampling-ell-r-from-P-r-1}
    For all $i \in [\RWSteps]$ if $\Robust \leq \RMDist/8$, then
    $\Pr[\mbox{$\ell_i$ has at least $2\Robust \RSlen$ non-zeros} \mid \wedge_{j=1}^{i-1} F_j]
     \geq \left( 1 - \frac{4}{|\F|} \right)$.
\end{claim}
\begin{proof}
The proof is rather immediate from \cref{lem:sampling}.
Let $\P_{i-1} = \Plane{\vecx_{i-1}}{h_{i-1}}{h_{i-1}'}$ be the plane chosen by \Cref{alg:H-CTRW} in the iteration $i-1$.
Note that conditioning on $(\wedge_{j=1}^{i-1} F_j)$ implies that $f_{\mid \P_{i-1}}$ has at least $(\RMDist-2\Robust)\RSlen^2$ non-zeros.%
\footnote{This follows only from conditioning on $F_{i-1}$, and the other $F_j$'s are irrelevant.}

Since $\ell_i$ is a uniformly random line in $\P_{i-1}$, by \cref{lem:sampling} it follows that
\begin{align*}
    \Pr[\mbox{$f_{|\ell_i}$ has less than $2\Robust \RSlen$ non-zeros}]
    &\leq \frac{1}{\abs{\F}} \cdot \frac{\RMDist-2\Robust}{(\RMDist-2\Robust-2\Robust)^2} \\
    &\leq \frac{4}{\abs{\F}}
    \enspace,
\end{align*}
where the last inequality is by the assumption that $\Robust \leq \RMDist/8$ and $\RMDist \geq 3/4$.
This completes the proof of \cref{claim:sampling-ell-r-from-P-r-1}.
\end{proof}
By applying \cref{claim:sampling-ell-r-from-P-r-1} we get
\begin{align}
\Pr[(\wedge_{j=1}^{i-1} F_j) \wedge  \mbox{$\ell_i$ has at least $2\Robust \RSlen$ non-zeros}]
&= \Pr[(\wedge_{j=1}^{i-1} F_j)] \cdot \Pr[\mbox{$\ell_i$ has at least $2\Robust \RSlen$ non-zeros} \mid \wedge_{j=1}^{i-1} F_j] \nonumber \\
&\geq \Pr[(\wedge_{j=1}^{i-1} F_j)] \cdot \left( 1 - \frac{4}{|\F|} \right)
\enspace. \label{eq:error-2-over-H}
\end{align}
Combining \cref{eq:equality} with \cref{eq:error-2-over-H} together we get
\begin{equation}\label{eq:Pr[F_i's]-peel-off}
    \Pr[(\wedge_{j=1}^{i} F_j)]
    \geq \Pr[(\wedge_{j=1}^{i-1} F_j)] \cdot \left( 1 - \frac{4}{|\F|} \right) - \eps_i
    \enspace.
\end{equation}
By exactly the same argument, using the assumption that $\Word_{\mid \P_0}$ contains at least $(\RMDist - 2\Robust)n^2$ non-zeros,
it follows that
\begin{equation}\label{eq:Pr[F-1]}
    \Pr[F_1] \geq 1 - \frac{4}{|\F|} - \eps_1
    \enspace.
\end{equation}
The rest of the proof follows by induction, peeling off one $F_i$ at a time,
and applying \cref{eq:Pr[F_i's]-peel-off}.
\begin{align*}
    \Pr[(\wedge_{i=1}^{\RWSteps} F_i)]
    &\geq \Pr[(\wedge_{i=1}^{\RWSteps-1} F_i)] \cdot \left( 1 - \frac{4}{|\F|} \right) - \eps_\RWSteps \\
    &\geq \Bigg(\Pr[(\wedge_{i=1}^{\RWSteps-2} F_i)] \cdot \left( 1 - \frac{4}{|\F|} \right) - \eps_{\RWSteps-1} \Bigg)\left( 1 - \frac{4}{|\F|} \right) - \eps_\RWSteps \\
    &= \Pr[(\wedge_{i=1}^{\RWSteps-2} F_i)] \cdot \left(1 - \frac{4}{|\F|}\right)^2 -\left(1 - \frac{4}{|\F|}\right)\epsilon_{\RWSteps - 1} - \epsilon_{\RWSteps} \\
    &\geq \dots \\
    &\geq \Pr[F_1] \cdot \left(1 - \frac{4}{|\F|}\right)^{\RWSteps-1} - \sum_{i=1}^{\RWSteps}{\left(1 - \frac{4}{|\F|}\right)^{\RWSteps-i}\cdot \epsilon_i} \\
    &\geq \left(1 - \frac{4}{|\F|} - \epsilon_1 \right) \cdot \left(1 - \frac{4}{|\F|} \right)^{\RWSteps-1} - \sum_{i=1}^{\RWSteps}{\left(1 - \frac{4}{|\F|}\right)^{\RWSteps-i} \cdot \epsilon_i} \\
    &= (1 - \frac{4}{|\F|})^{\RWSteps} - \sum_{i=1}^{\RWSteps}{\left(1 - \frac{4}{|\F|} \right)^{\RWSteps-i} \cdot \epsilon_i} \\
    &> (1 - \frac{4}{|\F|})^{\RWSteps} - \sum_{i=1}^{\RWSteps} \epsilon_i
    \enspace.
\end{align*}
We get that $\Pr[(\wedge_{i=1}^{\RWSteps} F_i)] > (1 - \frac{4}{|\F|})^{\RWSteps} - \sum_{i=1}^{\RWSteps} \epsilon_i$, which concludes
\cref{lemma:all-F-i-s-hold-whp}.
\end{proof}
\doclearpage
\section{PCPs of proximity}
\label{sec:PCPP}

In this section we explain how to construct PCPP systems for the languages $\RM_{\mid \P}^{(\vecx)}$ and $\RM_{\mid \P}^{(\ell)}$
defined in \cref{def:RM_P^x,,def:RM_P^ell}.
Note that since all planes in $\Field^\RMDim$ are isomorphic, we may think of each
$\RM_{\mid \P}^{(\vecx)}$ and $\RM_{\mid \P}^{(\ell)}$ as the language $\RM_\Field(2,\RMDeg)$ of bivariate polynomials of total degree at most $\RMDeg$,
concatenated with repetitions of their values at $\vecx$ and $\ell$ respectively.
Following \cref{notation:f_P^x} and \cref{notation:f_P^ell},
we make the following definition.

\begin{definition}\label{def:concat-repetition}
Let $\F$ be a field of size $\abs{\F} = \RMLen$,
and let $\Word \colon \F^2 \to \F$ be an $\F$-valued function.
Let $\Point \in \F^2$ be a point in $\F^2$,
$\ell \seq \Field^2$ be the line $\ell = \Line{\Point_1}{\Point_2} = \{ \Point_1 + t \cdot \Point_2 \colon t \in \F \}$ for some $\Point_1, \Point_2 \in \F^2$.
\begin{itemize}
    \item Define $\Word^{(\Point)}$ to be the concatenation of $\Word$ with $\RMLen^2$ repetitions of the value of $\Word$ at the point $\Point$, i.e., $\Word^{(\Point)} = \Word \circ (\Word(\Point))^{\RMLen^2}$.
    \item Define $\Word^{(\ell)}$ to be the concatenation of $\Word$ with $\RMLen$ repetitions of the restriction of $\Word$ to line $\ell$, i.e., $\Word^{(\ell)} = \Word \circ (\Word_{\mid \ell})^\RMLen$.
\end{itemize}

\noindent
Define $\RM^{(\Point)} = \RM^{(\Point)}_{\Field}(2,\RMDeg) = \{ Q^{(\Point)} \colon Q \in \RM_{\F}(2,\deg)\}$
and $\RM^{(\ell)} = \RM^{(\ell)}_{\Field}(2,\RMDeg) = \{ Q^{(\ell)} \colon Q \in \RM_{\F}(2,\deg)\}$.
\end{definition}

Note that given an oracle access to $\Word$ we can query every coordinate of $\Word^{(\Point)}$ and $\Word^{(\ell)}$
by querying one coordinate of $\Word$. In particular, any PCPP system for $\Word^{(\Point)}$
can be emulated when given access to $\Word$ without increasing the query complexity of the proof system.

The following observation is immediate from \cref{def:concat-repetition}.
\begin{observation}\label{obs:language-dist}
Let $\F$ be a field of size $\abs{\F} = \RMLen$,
Let $\Point \in \F^2$ be a point in $\F^2$,
$\ell \seq \Field^2$ be the line $\ell = \Line{\Point_1}{\Point_2} = \{ \Point_1 + t \cdot \Point_2 \colon t \in \F \}$ for some $\Point_1, \Point_2 \in \F^2$.
Then for any $\Word \colon \F^2 \to \F$ we have
\begin{equation*}
    \dist_{\Point}(\Word, \RM) = \dist(\Word^{(\Point)}, \RM^{(\Point)})
\quad
\mbox{and}
\quad
    \dist_{\ell}(\Word, \RM) = \dist(\Word^{(\ell)}, \RM^{(\ell)})
    \enspace.
\end{equation*}
\end{observation}

It is clear that the languages $\RM^{(\ell)}$ and $\RM^{(\Point)}$ can be solved in polynomial time.
Therefore, by Theorem 6.1 in \cite{CGS20}, these languages admit a ccPCPP with the appropriate parameters.

\begin{theorem}[Canonical PCPP for $\RM$]\label{thm:cPCPP-for-RM}
Let $\Field$ be a finite field and let $\RMDeg \in \N$ be a parameter.
Let $\Language$ be either $\RM^{(\ell)}_{\Field}(2,\RMDeg)$ or $\RM^{(\Point)}_{\F}(2,\RMDeg)$.
Then $\Language$ admits a ccPCPP with the following parameters
\begin{enumerate}
  \item The ccPCPP verifier has perfect completeness and soundness $\Soundness_{PCPP} = 0.5$
  for any proximity parameter $\Robust>0$.
  \item The query complexity of the verifier is $\Query = O(1/\Robust)$.
  \item The length of canonical proof $\pi(f)$ for $f \in \Language$ of length $n$ is $\PCPPLen(n) = \poly(n)$.
  \item The language $\Pi_{\Language} = \{f \circ \pi(f) : f \in \Language\}$ is a $(\DecRad,\SoundnessInnerRLCC)$-RLCC
    with query complexity $\Query = O(1)$,
    constant correction radius $\DecRad = \Omega(1)$,
    and constant soundness $\Soundness = \Omega(1)$.
\end{enumerate}
\end{theorem}

Informally speaking, in order to prove \cref{thm:cPCPP-for-RM} \cite{CGS20} start with a cPCPP system from \cref{thm:PCPP},
and for every $w \in \Language$, and a canonical proof $\pi(w)$,
define a correctable proof $\pi^*(w)$ by encoding $w \circ \pi(w)$ using a systematic RLCC with constant distance and polynomial block length
(e.g., the one from \cite{GurRR18} or \cite{CGS20}).
Since the RLCC is systematic, the encoding is of the form $w \circ \pi(w) \circ \pi'(w)$
for some string $\pi'(w)$ of length $\poly(\abs{w})$.
Then, define the canonical proof to be $\pi^*(w) = \pi(w) \circ \pi'(w)$.
It is rather straightforward to define a verifier that will satisfy the requirements of \cref{thm:cPCPP-for-RM}.
We omit the details from here, and refer the interested reader to Theorem 6.1 in \cite{CGS20}.

\doclearpage
\section{Composition theorem and the local correcting algorithm}
\label{sec:composition-local-alg}
In this section we present a composition theorem used to combine the $\CTRW$ for Reed-Muller codes with appropriate PCPPs.
This composition theorem immediately implies the statement of \cref{thm:main}, albeit for a large alphabet.

\subsection{Composition theorem using $\CTRW$}
Below we prove that if $\RM_\Field(\RMDim,\RMDeg)$ admits an $\RMDim$-steps $\CTRW$,
then it can be composed with a PCPP system with appropriate parameters to obtain an RLCC with query complexity $O(\RMDim)$.
The composition theorem that we present is a slightly modified version of the composition theorem presented in \cite{CGS20}.
The main difference compared to \cite{CGS20} is that we consider two types of PCPP proofs for $\RM$.

\begin{theorem}[Composition theorem for Reed-Muller codes]
  \label{thm:composition-rw}
  Let $\Field$ be a finite field of size $\abs{\Field} = \RMLen$, let $\RMDim,\RMDeg \in \N$ be parameters,
  and let $\RMDist = 1-\frac{\RMDeg}{\RMLen}$.
  Suppose that $\H$ is a subfield of $\Field$ such that $[\Field:\H] = \RMDim$.
  Consider the following components.
  \begin{itemize}
    \item Reed-Muller code $\RM_{\Field}(\RMDim,\RMDeg) \colon \Field^\CodeDim \to \Field^{\RMLen^\RMDim}$
    that admits an $\RMDim$-steps $\H$-plane-line-$\CTRW$ with the following parameters.
    \begin{enumerate}
    \item $\CTRW$ has perfect completeness and $(\TestRadius, \Robust, \SoundnessRW)$-robust soundness.
    \item The total number of predicates (of both types) defined for the $\CTRW$ is at most $B$.
    \end{enumerate}

    \item Canonical PCPP systems for languages of the form $\RM_{\mid \P}^{(\vecx)}$ and $\RM_{\mid \P}^{(\ell)}$ with the following properties.
    \begin{enumerate}
      \item For each $\Word_{\mid \P}^{(\vecx)} \in \RM_{\mid \P}^{(\vecx)}$ the length of the canonical proof is at most $\PCPPLen(2\RMLen^2)$.
      \item For each $\Word_{\mid \P}^{(\ell)} \in \RM_{\mid \P}^{(\ell)}$ the length of the canonical proof is at most $\PCPPLen(2\RMLen^2)$.
      \item The verifier has query complexity $\QueryPCPP$,
            perfect completeness, soundness $\SoundnessPCPP < 1$ for proximity parameter $\Robust$.
      \item The codes $\Pi_{\RM_{\mid \P}^{(\vecx)}} = \{\Word_{\mid \P}^{(\vecx)} \circ \Proof(\Word_{\mid \P}^{(\vecx)}) : \Word_{\mid \P}^{(\vecx)} \in \RM_{\mid \P}^{(\vecx)}\}$
      and
      $\Pi_{\RM_{\mid \P}^{(\ell)}} = \{\Word_{\mid \P}^{(\ell)} \circ \Proof(\Word_{\mid \P}^{(\ell)}) : \Word_{\mid \P}^{(\ell)} \in \RM_{\mid \P}^{(\ell)}\}$
      are $(2\Robust,\SoundnessInnerRLCC)$-RLCC with query complexity $\QueryPCPP$.
    \end{enumerate}

  \end{itemize}
  Then, there exists a code $\ComposedCode \colon \F^{\CodeDim} \to \F^{\BlockLength}$
  with block length $\BlockLength \leq \RMLen^\RMDim + 2 B \cdot \PCPPLen(2\RMLen^2)$
  and relative distance at least $\frac{1}{2} \left(1 - \frac{\RMDeg}{\RMLen}\right)$.

  The code $\ComposedCode$ is a $(\DecRad,\SoundnessRLCC)$-RLCC with query complexity $\QueryRLCC = (\RMDim + 3) \cdot \QueryPCPP$,
  where the decoding radius of $\ComposedCode$ is $\DecRad = \TestRadius/4$
  and the soundness is
  \begin{equation*}
    \SoundnessRLCC = \min   \left(
                                \frac{\SoundnessRW \cdot (1-\SoundnessPCPP) \cdot \RMDist}{2},
                                \frac{\SoundnessInnerRLCC}{2}
                            \right)
    \enspace.
  \end{equation*}
\end{theorem}

Before proceeding with the proof of \cref{thm:composition-rw}, we show how it implies \cref{thm:main}.

\begin{proof}[Proof of \cref{thm:main}]
Given a sufficiently large parameter $\QueryRLCC$ specifying the desired query complexity of an RLCC,
let $\RMDim = \floor{\frac{\QueryRLCC}{\QueryPCPP}}-1 \geq 2$, and let $\RMDeg \geq 16^\RMDim$.
Choose a prime field $\H$, such that $(4 \RMDeg)^{1/\RMDim} \leq \abs{\H} \leq 2 \cdot (4 \RMDeg)^{1/\RMDim}$.
Finally, we let $\Field$ be the degree-$\RMDim$ extension of $\H$,
and let $\RMLen = \abs{\Field}$.
In particular, by the choice of $\RMDeg$ we have $\abs{\H} \geq 16$,
and the relative distance of $\RM_{\Field}(\RMDim, \RMDeg)$
is $\RMDist = 1 - \frac{\RMDeg}{\RMLen} = 1 - \frac{\RMDeg}{\abs{\H}^\RMDim} \geq 1 - \frac{\RMDeg}{4 \RMDeg} \geq 3/4$.

By \cref{thm:H-ctrw}, the $\RM_{\Field}(\RMDim,\RMDeg)$ admits
an $\RMDim$-steps $\H$-plane-line-$\CTRW$ with perfect completeness
and $(\TestRadius, \Robust, \SoundnessRW)$-robust soundness,
with $\TestRadius = \RMDist/2$, $\Robust = \RMDist/8$,
and soundness parameter
\begin{equation*}
\SoundnessRW
= \left(1 - \frac{4}{\abs{\Field}} \right)^{\RMDim} - \frac{\TestRadius + \frac{2}{\abs{\SubF}}}{\RMDist-2\Robust}
> 1 - \frac{4\RMDim}{\abs{\Field}} - \frac{\delta/2 + 1/8}{3\delta/4}
\geq 1 - \frac{4\RMDim}{4\cdot 16^\RMDim} - \frac{3/8 + 1/8}{9/16}
\geq \frac{1}{16} - \frac{m}{16^m}
= \Omega(1)
\enspace.
\end{equation*}
Furthermore, by a simple counting, the total number of predicates (of both types) defined for the $\CTRW$
is $B \leq 2\RMLen^\RMDim \cdot \abs{\H}^{2\RMDim} \cdot \RMLen^2 = 2\RMLen^{\RMDim + 4}$.

For the canonical PCPP component, by \cref{thm:cPCPP-for-RM} the ccPCPP system has
perfect completeness, constant soundness with respect to $\Robust = \RMDist/8$,
query complexity $\QueryPCPP = O(1/\Robust) = O(1/\RMDist) = O(1)$,
and the length of the canonical proof is $\PCPPLen(2\RMLen^2) = \poly(\RMLen)$.

Therefore, by \cref{thm:composition-rw} we obtain a $\QueryRLCC$-query
$(\DecRad,\SoundnessRLCC)$-RLCC
$\ComposedCode \colon \F^{\CodeDim} \to \F^{\BlockLength}$
with constant relative distance, $\DecRad = \Omega(1)$, $\SoundnessRLCC = \Omega(1)$,
where the message length of the code is $\CodeDim = {\RMDeg + \RMDim \choose \RMDim} \geq \left(\frac{\RMDeg}{\RMDim}\right)^\RMDim$,
and its block length is
$\BlockLength \leq \CodeLen^\RMDim + 2 B \cdot \PCPPLen(2\CodeLen^2)
\leq \CodeLen^\RMDim + 4\CodeLen^{\RMDim + 4} \cdot \poly(\CodeLen)$.
By plugging in the parameters we get
\begin{equation*}
    \BlockLength
    = \CodeLen^{\RMDim+O(1)}
    \leq (2^\RMDim \cdot 4\RMDeg)^{\RMDim + O(1)}
    = (4 \RMDim \cdot 2^\RMDim)^{\RMDim + O(1)} \cdot \left(\frac{\RMDeg}{\RMDim}\right)^{\RMDim + O(1)}
    = 2^{O(\RMDim^2)} \cdot \CodeDim^{1 + O(1/\RMDim)}
    \enspace,
\end{equation*}
and relative distance is at least $\frac{1}{2} \left(1 - \frac{\RMDeg}{\RMLen}\right) \geq 3/8$.
This completes the proof of \cref{thm:main}.
\end{proof}

The rest of this section is devoting to the proof of \cref{thm:composition-rw}.

\subsection{Constructing the composed code}
\label{sec:composition-construction}
\parhead{Constructing the composed code $\ComposedCode$:}
Given the components in the statement of \cref{thm:composition-rw}, the composed code $\ComposedCode \colon \F^\CodeDim \to \F^{\BlockLength}$
is obtained by concatenating several repetitions of the Reed-Muller encoding of the message with the canonical proofs of proximity.

Specifically, given a message $M \in \Field^\CodeDim$, we first let $\codeword_{\RM} = \RM(M)$ be encoding of $M$ using $\BaseCode$.
The final encoding $\ComposedCode(M)$ consists of the following three parts:
\begin{equation*}
    \ComposedCode(\message) = \RMPart \circ \ProofPart_{Point} \circ \ProofPart_{Line}
    \enspace,
\end{equation*}
descried below.
\begin{enumerate}
    \item \label{desc:composed-code-base} $\RMPart$ consists of
        $\NumRep = \ceil{B \cdot \PCPPLen(2\CodeLen^2) / \CodeLen^\RMDim}$ repetitions of $\codeword_{\RM}$,
        where $\NumRep \geq 1$ is the minimal integer so that $\NumRep \cdot \CodeLen^\RMDim \geq B \cdot \PCPPLen(2\CodeLen^2)$.
        Although these repetitions look rather artificial, they make sure that the Reed-Muller part of the encoding
        will constitute a constant fraction of the codeword $\ComposedCode(\message)$.

    \item \label{desc:composed-code-proofs-P-x}
        $\ProofPart_{Point}$ is the concatenation of proofs of proximity $\Proof_{(\P, \Point)}$
        (as per the ccPCPPs in the hypothesis of the theorem) for each $\H$-plane $\P$
        and for each point $\Point \in \P$.
        That is, each such $\Proof_{(\P, \Point)}$ is the canonical proof for the
        assertion that $(\codeword_{\RM})_{\mid \P}^{(\Point)} \in \RM_{\mid \P}^{(\Point)}$.

        Note that since $\ComposedCode(\message)$ contains many copies of $\codeword_{\RM}$,
        each $\Proof_{(\P, \Point)}$ is expected to be the canonical proof for all the copies.

    \item \label{desc:composed-code-proofs-P-ell}
        $\ProofPart_{Line}$ is the concatenation of proofs of proximity $\Proof_{(\P, \ell)}$
        (as per the ccPCPPs in the hypothesis of the theorem) for each $\H$-plane $\P$
        and for each line $\ell \seq \P$.
        Each such $\Proof_{(\P, \ell)}$ is the canonical proof for the
        assertion that $(\codeword_{\RM})_{\mid \P}^{(\ell)} \in \RM_{\mid \P}^{(\ell)}$.

        Again, since $\ComposedCode(\message)$ contains many copies of $\codeword_{\RM}$,
        each $\Proof_{(\P, \ell)}$ is expected to be the canonical proof for all the copies.
\end{enumerate}

\parhead{Parameters of $\ComposedCode$:}

Note that the total block length of the encoding is
$\BlockLength = \NumRep \cdot \CodeLen^\RMDim + B \cdot \PCPPLen(2\CodeLen^2)
\leq \CodeLen^\RMDim + 2B \cdot \PCPPLen(2\CodeLen^2)$.

As for the relative distance of $\ComposedCode$,
if the relative distance of $\RM_\Field(\RMDim,\RMDeg)$ is $\RMDist$,
then, the relative distance of the $\RMPart$ part is also $\RMDist$.
Furthermore, since the length of the $\RMPart$ part is at least half of the total block length,
it follows that the relative distance of $\ComposedCode$ is at least $\RMDist/2$.

\subsection{Local correction algorithm for the Reed-Muller part}
\label{sec:correction-alg-RM}
Below we present the local correcting algorithm for $\ComposedCode$ for the $\RMPart$ part of the code.
Given a word $\inputword \in \F^{\BlockLength}$
write $\inputword = \Word^{rep} \circ \ProofPart_{Point} \circ \ProofPart_{Line}$,
where $\Word^{rep}$ is (expected to be) the $\NumRep$ copies of some Reed-Muller codeword,
and $\ProofPart_{Point}, \ProofPart_{Line}$ are the proofs as described above.

Let $i \in [\BlockLength]$ be the coordinate in the $\RMPart$ part of the code,
which corresponds to some $\Index \in \Field^\RMDim$ of (one of the copies of) the Reed-Muller encoding.
The local correcting algorithm is described in \cref{alg:local-correction-RM}, and works as follows.

\medskip

\begin{algorithm}[H]
    \caption{Local correcting algorithm for the $\RMPart$ part}\label{alg:local-correction-RM}
    \SetKwFunction{proc}{$\Subroutine$}
        \KwIn{$\inputword = \Word^{rep} \circ \ProofPart_{Point} \circ \ProofPart_{Line}, i \in [\BlockLength]$}

        Let $\Index \in \Field^\RMDim$ be the index corresponding to the $i$'th coordinate of the $\RM$ encoding

        Sample $\Randomness \in [\NumRep]$ uniformly at random,
        and let $\Word \colon \Field^\RMDim \to \Field$ be the substring of $\Word^{rep}$
        corresponding to the $\Randomness$-th copy of the base codeword \label{alg:correcting-RM-sample-f}

        Run the $\RMDim$-steps $\H$-Line-Plane-$\CTRW$ from \cref{alg:H-CTRW} on the input $(\Word, \Index)$

        Let $\P_0,\P_1, \dots,\P_{\RMDim}$ be the planes sampled by $\CTRW$,
        and let $\ell_1,\dots,\ell_{\RMDim}$ be the sampled lines \label{alg:local-correction-RM:constraints}

        Run the cPCPP verifier on $\Proof_{(\P_0, \Index)}$ to check that
        $\Word_{\mid \P_0}^{(\Index)}$ is $\Robust$-close to $\RM_{\mid \P_0}^{(\Index)}$
        \label{alg:decoding-RM-PCPP-plane-point}

        \For{$j = 1$ \KwTo $\RMDim$}{
            Run the cPCPP verifier on $\Proof_{(\P_j, \ell_j)}$ to check that
            $\Word_{\mid \P_j}^{(\ell_j)}$ is $\Robust$-close to $\RM_{\mid \P_j}^{(\ell_j)}$
            \label{alg:decoding-RM-PCPP-plane-line}
        }

        \If {Step \ref{alg:decoding-RM-PCPP-plane-point} accepts and all iterations of Step \ref{alg:decoding-RM-PCPP-plane-line} accept}{
            \Return $\Word(\Index)$
        }
        \Else{
            \Return $\bot$
        }
\end{algorithm}

\parhead{Query complexity:}
The total number of queries made in \cref{alg:local-correction-RM} is clearly
upper bounded by $(\RMDim+1) \cdot \QueryPCPP$
from Lines \ref{alg:decoding-RM-PCPP-plane-point} and \ref{alg:decoding-RM-PCPP-plane-line} of the algorithm.

\parhead{Proof of correctness:}
By the description of the algorithm, it is clear that if the input is a non-corrupted codeword,
i.e., $\inputword \in \ComposedCode$ then for any $\Index^* \in \Field^\RMDim$,
the algorithm always returns the correct answer.

Now, assume that the input
$\inputword = \Word^{rep} \circ \ProofPart_{Point} \circ \ProofPart_{Line}$ is $\DecRad$-close to some codeword $\ClosestCodeword \in \ComposedCode$,
and suppose that the $\RMPart$ part of $\ClosestCodeword$ consists of $\NumRep$ copies of some degree-$\RMDeg$ polynomial $\codeword^* \in \RM$.
We will show that $\Pr[\Decoder_{\cref{alg:local-correction-RM}}^{\inputword}(i) \in \{\codeword^*(\Index), \bot\}] \geq \SoundnessRLCC$.

Note that since $\inputword$ is $\DecRad$-close to $\ClosestCodeword$,
and the length of $\Word^{rep}$ is at least 1/2 of the total block length,
it follows that $\Word^{rep}$ is $2\DecRad$-close to the $\NumRep$ repetitions of $\codeword^*$.
Denote $\Wclose$ to be the event that $\dist(\Word, \codeword^*) \leq 4\DecRad = \TestRadius$
for the random copy $\Word$ in the $\RMPart$ part sampled in Line \ref{alg:correcting-RM-sample-f} of the algorithm.
Then, by Markov's inequality $\Pr[\Wclose] \geq 1/2$.
Therefore,
\begin{equation*}
        \Pr[\Decoder_{\cref{alg:local-correction-RM}}^{\inputword}(i) \in \{\codeword^*(\Index), \bot\} ]
        \geq 1/2 \cdot \Pr[ \Decoder_{\cref{alg:local-correction-RM}}^{\inputword}(i) \in \{\codeword^*(\Index), \bot\} | \Wclose]
      \enspace.
\end{equation*}

From now on, let us condition on the event $\Wclose$,
and focus on the term $\Pr[ \Decoder_{\cref{alg:local-correction-RM}}^{\inputword}(i) \in \{\codeword^*(\Index), \bot\} | \Wclose]$.
Furthermore, let us fix the choice of $\Word$ and consider the following two cases.
\begin{description}
  \item[Case 1: $\Word(\Index) = \codeword^*(\Index)$.] Noting that $\Decoder_{\cref{alg:local-correction-RM}}^{\inputword}(i)$
  always outputs either $\Word(\Index)$ or $\bot$ it follows that by conditioning on such $\Word$ we have
\begin{equation*}
    \Pr[ \Decoder_{\cref{alg:local-correction-RM}}^{\inputword}(i) \in \{\codeword^*(\Index), \bot\} | \Wclose, \Word(\Index) = \codeword^*(\Index)] = 1
    \enspace.
\end{equation*}

  \item[Case 2: $\Word(\Index) \neq \codeword^*(\Index)$.] In this case
  since $\CTRW$ admits $(\TestRadius, \Robust, \SoundnessRW)$-robust soundness,
  it follows that
        \begin{equation*}
            \Pr[\textrm{$\dist_{\Index}(\Word_{\mid \P_0}, \RM_{\mid \P_0}) \geq \Robust
            \vee
            \exists j \in [\RMDim]$ such that $\dist_{\ell_j}(\Word_{\mid \P_j}, \RM_{\mid \P_j}) \geq \Robust$}]
            \geq \SoundnessRW
            \enspace.
        \end{equation*}
\end{description}
Therefore, when running the cPCPP verifier for which the local view is $\Robust$-far from the corresponding predicate,
the verifier will reject with probability at least $1 - \SoundnessPCPP$,
and hence the decoder will output $\bot$ with the same probability.
Therefore, we can lower bound the second term by
\begin{equation*}
    \Pr[ \Decoder_{\cref{alg:local-correction-RM}}^{\inputword}(i) \in \{\codeword^*(\Index), \bot\} | \Wclose, \Word(\Index) \neq \codeword^*(\Index)]
    \geq \SoundnessRW (1-\SoundnessPCPP)
    \enspace.
\end{equation*}
Putting all together, we conclude
\begin{equation*}
        \Pr[\Decoder_{\cref{alg:local-correction-RM}}^{\inputword}(i) \in \{\codeword^*(\Index), \bot\}
        \geq \frac{\SoundnessRW \cdot (1-\SoundnessPCPP)}{2} \geq \SoundnessRLCC
        \enspace.
\end{equation*}
This completes the proof of correctness of the algorithm for the $\RMPart$ part of the code.

\subsection{Local correction algorithm for the proof part}
\label{sec:correction-alg-PCPP}
Next, we present the correction algorithm for the cPCPP proofs part of the code.
Let $\inputword = \Word^{rep} \circ \ProofPart_{Point} \circ \ProofPart_{Line} \in \Field^\BlockLength$
be a given word, and let $i \in [\BlockLength]$ be a coordinate
in the $\ProofPart_{Point} \circ \ProofPart_{Line}$ part of the proof.
The correction algorithm is described in \cref{alg:local-correction-PCPP}, and works as follows.

\begin{algorithm}[H]
\caption{Local correction for the PCPP part $\ProofPart$}\label{alg:local-correction-PCPP}
    \KwIn{$\inputword = \Word^{rep} \circ \ProofPart_{Point} \circ \ProofPart_{Line}, i \in [\BlockLength]$}

    Sample $\Randomness \in [\NumRep]$ uniformly at random,
    and let $\Word \colon \Field^\RMDim \to \Field$ be the substring of $\Word^{rep}$
    corresponding to the $\Randomness$-th copy of the base codeword \label{alg:RLCC-i-pcpp-sample-f}

    \If {$i$ is a coordinate in $\ProofPart_{Point}$}{
        Let $\P^\star$ and $\Index^\star \in \P^\star$ be the plane and the point such that $i$ is a coordinate of $\Proof = \Proof_{(\P^\star,\Index^\star)}$

        Run the cPCPP verifier to check that $\dist(\Word_{\mid \P^\star}^{(\Index^\star)}, \RM_{\mid \P^\star}^{(\Index^\star)}) \leq \Robust$ \label{alg:RLCC-i-pcpp-check-proof-x}

        \If {Step \ref{alg:RLCC-i-pcpp-check-proof-x} rejects}{
                \Return $\bot$ \label{alg:RLCC-i-pcpp-check-proof-x-ret}
        }
    }

    \Else{
        Let $\P^\star$ and $\ell^\star \seq \P^\star$ be the plane and the line such that $i$ is a coordinate of $\Proof = \Proof_{(\P^\star,\ell^\star)}$

        Run the cPCPP verifier to check that $\dist(\Word_{\mid \P^\star}^{(\ell^\star)}, \RM_{\mid \P^\star}^{(\ell^\star)}) \leq \Robust$ \label{alg:RLCC-i-pcpp-check-proof-ell}

        \If {Step \ref{alg:RLCC-i-pcpp-check-proof-ell} rejects}{
                \Return $\bot$ \label{alg:RLCC-i-pcpp-check-proof-ell-ret}
        }
    }

        Choose a uniformly random $\Index_0 \in \P^\star$ 

        Run the $\RMDim$-steps $\H$-Line-Plane-$\CTRW$ from \cref{alg:H-CTRW} on the input $(\Word, \Index_0)$ \label{alg:RLCC-i-pcpp-run-H-CTRW}

        Let $\P_0,\P_1, \dots,\P_{\RMDim}$ be the planes sampled by $\CTRW$,
        and let $\ell_1,\dots,\ell_{\RMDim}$ be the sampled lines

        Run the cPCPP verifier on $\Proof_{(\P_0, \Index_0)}$ to check that
        $\Word_{\mid \P_0}^{(\Index_0)}$ is $\Robust$-close to $\RM_{\mid \P_0}^{(\Index_0)}$
        \label{alg:RLCC-i-pcpp-plane-point}

        \For{$j = 1$ \KwTo $\RMDim$}{
            Run the cPCPP verifier on $\Proof_{(\P_j, \ell_j)}$ to check that
            $\Word_{\mid \P_j}^{(\ell_j)}$ is $\Robust$-close to $\RM_{\mid \P_j}^{(\ell_j)}$
            \label{alg:RLCC-i-pcpp-plane-line}
        }

        \If {Steps \ref{alg:RLCC-i-pcpp-plane-point} or \ref{alg:RLCC-i-pcpp-plane-line} reject} {
                \Return $\bot$  \label{alg:RLCC-i-pcppl-RW-fails}
        }

    \If {$i$ is a coordinate in $\ProofPart_{Point}$} {
        Run the local corrector of the inner ccPCPP on $\Word_{\mid \P^\star}^{(\Index^\star)} \circ \Proof_{(\P^\star,\Index^\star)}$
        to correct $\inputword_i$ \label{alg:RLCC-i-pcpp-correct-x}

        \Return the value obtained in Step \ref{alg:RLCC-i-pcpp-correct-x} \label{alg:RLCC-i-pcpp:step:PCPP-x-corr-return}
    }
    \Else{
        Run the local corrector of the inner ccPCPP on $\Word_{\mid \P^\star}^{(\ell^\star)} \circ \Proof_{(\P^\star,\ell^\star)}$
        to correct $\inputword_i$ \label{alg:RLCC-i-pcpp-correct-ell}

        \Return the value obtained in Step \ref{alg:RLCC-i-pcpp-correct-ell} \label{alg:RLCC-i-pcpp:step:PCPP-ell-corr-return}
    }
\end{algorithm}

\parhead{Query complexity:}
The total number of queries is upper bounded by
\begin{inparaenum}[(i)]
\item $\QueryPCPP$ queries in Step \ref{alg:RLCC-i-pcpp-check-proof-x} or in Step \ref{alg:RLCC-i-pcpp-check-proof-ell},
\item at most $(\RMDim+1) \cdot \QueryPCPP$ queries in Steps \ref{alg:RLCC-i-pcpp-plane-point} and \ref{alg:RLCC-i-pcpp-plane-line},
and
\item at most $\QueryPCPP$ queries in Step~\ref{alg:RLCC-i-pcpp-correct-x} or Step \ref{alg:RLCC-i-pcpp-correct-ell} .
\end{inparaenum}
Therefore, the total query complexity is upper bounded by $(\RMDim+3) \cdot \QueryPCPP$, as required.

\parhead{Proof of correctness:}
By the description of the algorithm, it is clear that if $\inputword \in \ComposedCode$,
then for any index $i \in [\BlockLength]$ in the proof part, the algorithm always returns the correct answer $\inputword_i$.

We assume from now on that the input $\inputword \in \F^{\BlockLength}$ is $\DecRad$-close to some codeword $\ClosestCodeword \in \ComposedCode$,
and suppose that the $\RMPart$ part of $\ClosestCodeword$ consists of $\NumRep$ copies of some degree-$\RMDeg$ polynomial $\codeword^* \in \RM$.
As in the previous part, since $\inputword$ is $\DecRad$-close to $\ClosestCodeword$,
and the length of $\Word^{rep}$ is at least 1/2 of the total block length,
it follows that $\Word^{rep}$ is $2\DecRad$-close to the $\NumRep$ repetitions of $\codeword^*$.
Therefore, for the random copy $\Word$ in the $\RMPart$ part sampled in Line \ref{alg:RLCC-i-pcpp-sample-f} of the algorithm,
we have $\Pr[\dist(\Word, \codeword^*) \leq 4\DecRad = \TestRadius] \geq 1/2$.
From now on let us condition on the event $\dist(\Word, \codeword^*) \leq \TestRadius$.

\medskip

Let us assume that the coordinate $i \in \BlockLength$ we wish to decode belongs to some $\Proof_{(\P^\star,\Index^\star)}$.
The following claim completes the analysis of the correcting algorithm.

\begin{claim}
If $\Pr[\Decoder_{\cref{alg:local-correction-PCPP}}^{\inputword}(i) = \bot] < \SoundnessRLCC$,
then $\Pr[ \Decoder_{\cref{alg:local-correction-PCPP}}^{\inputword}(i) = \ClosestCodeword_i] > \frac{\SoundnessInnerRLCC}{2}$.
\end{claim}

\begin{proof}
    Since \cref{alg:local-correction-PCPP} returns $\bot$ with probability
    less than $\SoundnessRLCC$ in Step \ref{alg:RLCC-i-pcpp-check-proof-x-ret},
    the PCPP verifier for $\Proof_{(\P^\star, \Index^\star)}$ in Step \ref{alg:RLCC-i-pcpp-check-proof-x}
    accepts with probability at least $1-\SoundnessRLCC > \SoundnessPCPP$.
    Thus, there is some bivariate degree-$\RMDeg$ polynomial $\codeword' \colon \P^\star \to \Field$ (not necessarily equal to $\codeword^*_{\mid \P^\star}$)
    such that
    \begin{enumerate}
        \item $\dist(\Word_{\mid \P^\star}^{(\Index^\star)}, {\codeword'}_{\mid \P^\star}^{(\Index^\star)}) \leq \Robust$,
        and hence $\dist(\Word_{\mid \P^\star}, \codeword'_{\mid \P^\star} )\leq 2\Robust$,
        \item and $\Proof_{(\P^\star, \Index^\star)}$ is $\Robust$-close to the canonical proof $\Proof({\codeword'}_{\mid \P^\star}^{(\Index^\star)})$.
    \end{enumerate}

    Next, we use the assumption that \cref{alg:local-correction-PCPP} returns $\bot$ with probability less than $\SoundnessRLCC$ in Steps \ref{alg:RLCC-i-pcpp-plane-point} or \ref{alg:RLCC-i-pcpp-plane-line}.
    That is, when running $\H$-plane-line $\CTRW$ from a uniformly random $\Index_0 \in \P^\star$,
    and then running the corresponding cPCPPs,
    with probability at least $1-\SoundnessRLCC$ all cPCPP verifiers accept.
    For each $\vecz \in \P^\star$ let $p_{\vecz}$ be the probability that
    both Step~\ref{alg:RLCC-i-pcpp-plane-point} and Step~\ref{alg:RLCC-i-pcpp-plane-line}
    accept when starting from $\vecz$.
    Then $\E[p_{\vecz}] > 1 - \SoundnessRLCC$,
    and hence for at least $(1 - \DistRM/2)\CodeLen^2$ starting points $\vecz \in \P^\star$
    we have $p_{\vecz} \geq 1 - \frac{2\SoundnessRLCC}{\DistRM} \geq 1 - \SoundnessRW (1-\SoundnessPCPP)$.
    Therefore, by the analysis of the correcting algorithm for the $\RMPart$ part,
    for more than $(1 - \DistRM/2)\CodeLen^2$ starting points $\vecz \in \P^\star$
    it holds that $\Word(\vecz) = \codeword^*(\vecz)$.
    Indeed, by \emph{case 2} of the analysis if $\Word(\vecz) \neq \codeword^*(\vecz)$,
    then $p_{\vecz} < 1 - \SoundnessRW (1-\SoundnessPCPP)$.
    Therefore, $\dist(\Word_{\mid \P^\star}, \codeword^*_{\mid \P^\star}) < \DistRM/2$.

  Combining with the conclusion from the previous step that $\dist(\Word_{\mid \P^\star}, \codeword'_{\mid \P^\star} )\leq 2\Robust$
  it follows that $\dist(\codeword^*_{\mid \P^\star}, \codeword'_{\mid \P^\star}) < 2 \Robust + \DistRM/2 \leq \DistRM$.
  Thus, since $\RM_\Field(\RMDim,\RMDeg)$ has distance $\DistRM$ we conclude that $\codeword^*_{\mid \P}= \codeword'_{\mid \P}$.

  \medskip

  So far we showed that if \cref{alg:local-correction-PCPP} returns $\bot$ with probability less than $\SoundnessRLCC$
  and $\Word$ is $\TestRadius$-close to $\codeword^*$ (which happens with probability at least $1/2$),
  then $\Word_{\mid \P^\star}$ is $2\Robust$-close to $\codeword^*_{\mid \P^\star}$, and $\Proof_{(\P^\star, \Index^\star)}$
  is $\Robust$-close to $\Proof({\codeword^*}_{\mid \P^\star}^{(\Index^\star)})$, the canonical proof of ${\codeword^*}_{\mid \P^\star}^{(\Index^\star)}$.
  Therefore, the local correction algorithm for the inner ccPCPP applied on $(\Word_{\mid \P^\star} \circ \Proof_{(\P^\star, \Index^\star)})$
  in Step \ref{alg:RLCC-i-pcpp-correct-x} returns either $\ClosestCodeword_i$ or $\bot$ with probability at least $\SoundnessInnerRLCC$.
  Therefore,
\begin{align*}
        \Pr[ \Decoder_{\cref{alg:local-correction-PCPP}}^{\inputword}(i) \in \{\ClosestCodeword_{i}, \bot\}]
        &\geq \Pr[ \text{Step \ref{alg:RLCC-i-pcpp:step:PCPP-x-corr-return}
                returns $\ClosestCodeword_i$ or $\bot$} | \dist(\Word, \codeword^*) \leq \TestRadius]
         \cdot \Pr[\dist(\Word, \codeword^*) \leq \TestRadius] \\
        &\geq \frac{\SoundnessInnerRLCC}{2}
        \enspace,
  \end{align*}
  as required.
\end{proof}

We proved correctness of the local correction algorithm assuming that the coordinate $i \in \BlockLength$ we wish to decode belongs to some $\Proof_{(\P^\star,\Index^\star)}$.
For the case when $i$ belongs to $\Proof_{(\P^\star,\ell^\star)}$, the analysis is exactly the same.
This concludes the proof of \cref{thm:composition-rw}.

\doclearpage
\section{Concluding remarks and open problems}
\label{sec:conclusions-open-problems}

In this paper we constructed an $O(\Query)$-query RLDC
$\Code \colon \Field^\CodeDim \to \Field^\BlockLength$ with block length
$\BlockLength=\Query^{O(\Query^2)} \cdot \CodeDim^{1+O(1/\Query)}$,
assuming that the field is large enough, namely, assuming that $\abs{\Field} \geq c_q \cdot \CodeDim^{1/\Query}$.
Using standard techniques it is possible to obtain a binary RLDC with similar parameters.
This can be done by concatenating our code with an arbitrary binary code with constant rate and constant relative distance.
Indeed, this transformation appears in~\cite[Appendix A]{CGS20},
who showed how concatenating $\CTRW$-based RLDC over large alphabet with a good binary code
gives a binary RLDC that essentially inherits the block length and the query complexity of the RLDC over large alphabet.
Below we provide the proof sketch, explaining how the concatenation works.

\begin{proof}[Proof sketch]
Suppose that we want to construct a short binary RLCC.
Let $\Code_{RLCC}: \F^{\CodeDim} \to \F^{\BlockLength}$ be the RLCC over some field $\Field$ with the desired block length,
and let $\Code_{bin}: \Bits^{\CodeDim'} \to \Bits^{\BlockLength'}$ be an error-correcting code with constant rate and constant distance.
We also assume that field $\F$ is chosen so that $\abs{\F} = 2^{\CodeDim'}$.
(To satisfy this condition, one can simply set $\H$ to be a field of characteristic 2.)
This assumption will allow us to have a bijection between each symbol
of $\F$ and binary string of length $\CodeDim'$.

We construct the binary concatenated code $\Code_{concat}: \{0,1\}^{\CodeDim \cdot \CodeDim'} \to
\{0,1\}^{\BlockLength \cdot \BlockLength'}$ as follows. Given a message $\message \in \{0,1\}^{\CodeDim \cdot \CodeDim'}$, we first
convert it to an string in $\message' \in \F^{\CodeDim}$ in the natural way.
Then, we encode $\message'$ using $\Code_{RLCC}$ to obtain a codeword $c^* \in \Code_{RLCC}$. Finally, we encode each symbol of $c^*$
using $\Code_{bin}$ to get the final codeword $c \in \{0,1\}^{\BlockLength \cdot \BlockLength'}$.

To prove that the concatenated code is an RLCC, Chiesa, Gur, and Shinkar proved in \cite[Theorem A.4]{CGS20}
that if $\Code_{RLCC}$ admits an $\RWSteps$-steps $\CTRW$ with some soundness guarantees,
then $\Code_{concat}$ admits an $\RWSteps$-steps $\CTRW$ with related soundness guarantees.
The $\CTRW$ on the concatenated code $\Code_{concat}$ emulates the $\CTRW$ on $\Code_{RLCC}$
by sampling planes for the $\CTRW$ on the Reed-Muller code,
and instead of reading the symbols from $\Field$, it reads the binary encodings of all symbols belonging to these planes.

Indeed, it is not difficult to see that if $\Code_{RLCC}$ admits an $\RWSteps$-steps $\CTRW$ with some soundness guarantees,
then so does the concatenated code.
We omit the details, and refer the interested reader to Appendix~A in~\cite{CGS20}.
\end{proof}

We conclude the paper with several open problems we leave for future research.

\begin{enumerate}

  \item The most fundamental open problem regarding RLDCs/RLCCs is to understand the optimal trade-off between
  the query complexity of LDCs and their block length in the constant query regime.
  It is plausible that the lower bound of \cite{GurL19} can be improved to $\CodeDim^{1 + \Omega(1/\Query)}$,
  although we do not have any evidence for this.

  \item As discussed in the intoduction, \cite{BGHSV06} asked whether it is possible to prove a separation between LDCs and RLDCs.
  Understanding the trade-off between the query complexity and the block length is one possible way to show such separation.

  \item Another interesting open problem is to construct an RLDC/RLCC with constant rate and small query complexity.
  In particular, it is plausible that there exist $\polylog(\BlockLength)$-query RLDCs with $\BlockLength = O(\CodeDim)$.

  \item Also, it would be interesting to construct RLDCs/RLCCs using high-dimensional expanders \cite{KM17, DK17, DDFH18, KO18}.
  Since there are several definitions of high-dimensional expanders,
  it would be interesting to state the sufficient properties of high-dimensional expanders required for RLDCs.
  We believe this approach can be useful in constructing constant rate RLDCs with small query complexity.
\end{enumerate}

\bibliographystyle{alpha}
\bibliography{references}

\end{document}